\documentclass[12pt]{amsart}

\usepackage{macros-new,amsaddr,physics}

\begin{document}

\title[Superconformal algebras and holomorphic field theories]{Superconformal algebras \\ and holomorphic field theories \\[1em] (\lowercase{or:} T\lowercase{owards derived enhancements of superconformal algebras,} I)}

\author{Ingmar Saberi}
\address{Ludwig-Maximilians-Universit\"at M\"unchen \\ Fakult\"at f\"ur Physik \\ Theresienstra\ss{}e 37 \\ 80333 M\"unchen \\ Deutschland}
\email{i.saberi@physik.uni-muenchen.de}

\author{Brian R. Williams}
\address{School of Mathematics\\
University of Edinburgh \\ 
Edinburgh \\ 
UK}
\email{brian.williams@ed.ac.uk}

\begin{abstract}
We show that four-dimensional superconformal algebras admit an infinite-dimensional derived enhancement after performing a holomorphic twist. 
The type of higher symmetry algebras we find are closely related to algebras studied by Faonte--Hennion--Kapranov, Hennion--Kapranov, and the second author with Gwilliam in the context of holomorphic QFT. 
We show that these algebras are related to the two-dimensional chiral algebras extracted from four-dimensional superconformal theories by Beem and collaborators; further deforming by a superconformal element induces the Koszul resolution of a plane in~$\C^2 \cong \RR^4$. 
The central charges at the level of chiral algebras arise from central extensions of the higher symmetry algebras. 
\end{abstract}

\date{\today}

\maketitle

\newpage

\tableofcontents

\section{Introduction}

Superconformal field theories are an important and interesting class of supersymmetric quantum field theories, characterized, at least roughly speaking, by their lack of any characteristic energy scale.
Algebraically, superconformal theories admit an action of the superconformal super Lie algebra.
This algebra contains both the algebra of conformal transformations and the super Poincar\'e algebra.
The super Poincar\'e algebra $\lie{saff}(d,\N)$ is a super Lie algebra containing the algebra of infinitesimal affine symmetries on $\RR^d$, together with odd symmetries $\lie{saff}(d,\N)_-$ in spinor representations of~$\lie{so}(d)  \subset \lie{aff}(d)$. 
The integer (or pair of integers) $\N$ labels the number of copies of the minimal set of odd generators that are present, and so can be thought of as the degree of extended supersymmetry. 
Superconformal algebras should thus sit in a commuting diagram of inclusions of super Lie algebras:
\[
  \begin{tikzcd}
    & \lie{conf}(d) \ar[dr] & \\
    \lie{aff}(d) \ar[ru] \ar[rd] & & \lie{sconf}(d,\N). \\
    & \lie{saff}(d,\N) \ar[ru]
  \end{tikzcd}
\]
A superconformal algebra, if it exists, is a completion of this diagram which is a simple super Lie algebra, and for which 
\[
\dim \lie{sconf}(d,\N)_- = 2 \dim \lie{saff}(d,\N)_-.
\]
Unlike super Poincar\'e algebras, algebras describing superconformal symmetries cannot be constructed in every dimension; examples exist only up to spacetime dimension six. (The first classification of superconformal algebras is due to Nahm~\cite{Nahm}, making use of Kac's classification of simple super Lie algebras~\cite{Kac}.)

It is well known that the algebra of conformal vector fields in two dimensions is infinite-dimensional; as such, all two-dimensional superconformal algebras have the same property.
In spacetime dimensions from three to six, however, superconformal algebras are finite-dimensional.
The goal of this paper is to provide an infinite-dimensional enhancement of the four-dimensional superconformal algebra which exists after performing a {\em twist} of the supersymmetric field theory. 

At root, a twist of a supersymmetric field theory is obtained by taking the invariants of an appropriate fermionic element of the super Poincar\'e algebra. This generally means asking that the chosen supercharge be a Maurer--Cartan element, and therefore define a deformation of the differential; the twist is then precisely the corresponding deformation. 
Maurer--Cartan elements in super Poincar\'e algebras have been classified~\cite{NV,Chris}; since the internal differential is here trivial, the Maurer--Cartan equation reduces to the familiar condition $\{Q,Q\}=0$.

The most heavily studied examples of twists are topological. Such twists extract a topological quantum field theory from a supersymmetric theory as studied in physics. Such twists are of enormous interest, since topological quantum field theories are amenable to axiomatization and  provide invariants of manifolds. However, as tools for studying the full field theory, topological twists leave much to be desired: they are only available in the presence of sufficient extended supersymmetry, and forget much of the data of the supersymmetric theory from which they arose. 

The minimal twists are in fact not topological, but rather are \emph{holomorphic}. These have been studied by many authors over the last twenty-five years; we cite~\cite{JohansenHet,NekThesis,CostelloHol, Yangian} just for example. As tools for the study of the original supersymmetric theory, these have three distinct advantages over other twists: firstly, they are more often available, appearing in any even-dimensional theory for which nontrivial Maurer--Cartan elements are present.
For example, any four-dimensional supersymmetric theory admits a holomorphic twist. 
Secondly, the holomorphic twist is the least forgetful twist; the space of nilpotent supercharges is naturally stratified~\cite{NV}, and as such lives naturally over a poset. Holomorphic twists always form the minimal elements of this poset, and therefore can be used to study any other twist by further deformation. This also means that, even for theories that do admit topological twists---such as $\N=2$ theories in four dimensions---the holomorphic twist can be used to extract much finer information about the original theory.

The third and final point is that holomorphic theories have a richer and more intricate structure than topological theories, admitting (for example) nontrivial operator product expansions that depend holomorphically on the spacetime. The familiar example to keep in mind is the distinction between two-dimensional topological field theories (which, by a simple and familiar classification, correspond to finite-dimensional Frobenius algebras) and vertex algebras.  

Vertex algebras---and, relatedly, the familiar phenomenon of symmetry enhancement in two-dimensional chiral theories, which replaces finite-dimensional global or conformal symmetries by infinite-dimensional Kac--Moody or Virasoro algebras---have long been seen as peculiar to two-dimensional physics. One main philosophical point of this note is to argue that these phenomena, which have been of such enormous importance and profit to  theoretical physics at least since the foundational work of~\cite{BPZ}, occur in holomorphic theories much more generally, and are not at all peculiar to two dimensions \emph{per se}. 

The reason that attention has largely been restricted to two-dimensional theories thusfar has to do with two distinct phenomena. The first of these is that the wave equation, on which free field theory is based, factors into left- and right-moving (or holomorphic and anti-holomorphic) sectors. This means that ordinary field theory in two dimensions is very closely related to holomorphic field theory, even in the absence of supersymmetry. In higher dimensions, this of course fails; as outlined above, though, there is still a close connection between supersymmetric and holomorphic theories. 

The second, perhaps more subtle, reason is often alluded to in the physics literature by citing Hartogs' theorem, which implies that every holomorphic function on~$\C^n\setminus 0$ (for $n\geq 2$) extends to a holomorphic function on~$\C^n$. It thus seems to be hopeless to make sense of an analogue of the Kac--Moody construction in more than one complex dimension. Let us give a brief outline of the usual argument for this enhancement in two dimensions that shows how it seems to break down for $n\geq 2$.

Suppose a theory has a global symmetry by a Lie algebra $\fg$. The only obstruction to a local symmetry is the presence of derivatives in the kinetic term; as such, the holomorphic theory admits a symmetry by all holomorphic functions with values in~$\fg$, since only the $\dbar$ operator appears in the action functional. On the local operators, there is a symmetry by any holomorphic function on the punctured affine plane, 
\deq{
  \cO_\hol(\C^n\setminus 0) \otimes \lie{g} \quad \xto{n=1}  \quad\C[z,z^{-1}]\otimes\fg.
}
When $n=1$, there is then a central extension by the residue pairing, which gives rise to the Kac--Moody algebra and is represented in interesting fashion on the local operators. When $n>1$, there are no meromorphic functions and no such pairing on $\cO_\hol(\C^n\setminus 0)$ exists.

However, a natural analogue of the residue map \emph{does} exist. It is, however, not defined on~$\cO_\hol$ itself, but rather on its derived replacement: the Dolbeault complex $\Omega^{0,\bu}(\C^n\setminus 0)$. (It is  worth emphasizing here that a twist of a physical field theory will always produce such a derived replacement, since the original sheaves of fields or currents are locally free over $\C^\infty$ functions at the cochain level.) The homotopy type of $\C^n\setminus 0$ is, of course, that of the $(2n-1)$-sphere, and so the wedge product followed by integration over the top class defines a pairing on differential forms. 
The integration map is a trace of degree $n-1$ on the cdga of Dolbeault forms, defined by taking
\deq{
  \alpha \mapsto \int_{S^{2n-1}} \alpha\wedge \Omega,
}
where $\Omega$ is a holomorphic Calabi--Yau form on~$\C^n\setminus 0$. The degree of the pairing on Dolbeault cohomology is therefore $n-1$, which is zero precisely in complex dimension one. In general, Dolbeault cohomology of punctured affine space is supported in degrees $0$ and $n-1$, and can be thought of as consisting of holomorphic functions on affine space in degree zero, together with their dual (multiples of the Bochner--Martinelli kernel) in degree $n-1$. These are superimposed, purely by accident, in complex dimension one, and form the positive- and negative-degree parts of the Laurent polynomials $\C[z,z^{-1}]$.
Thus, in our view, the second confusing coincidence in complex dimension one is the fact that Dolbeault cohomology is supported only in degree zero in this case, and the residue map is defined without any shift of grading. 

At this stage, it is worth remarking on a connection between the structure at hand and recent work~\cite{SecondaryOps} studying higher operations in topological quantum field theory arising from the homology  of the operad of little $n$-disks (i.e., of configuration spaces of points in~$\R^n$). The ghost number in our higher algebras is essentially Dolbeault form degree, and a holomorphic analogue of topological descent is possible, making use of those supercharges which witness a nullhomotopy of the antiholomorphic translations in the twisted theory. The graded pairing that gives rise to higher central extensions, as we have emphasized, arises from the top class in the homology of~$\C^n\setminus 0$, which is the same class that gives rise to the bracket operation on local operators in TQFT discussed in~\cite{SecondaryOps}, albeit paired with the Calabi--Yau form. In a sense, for us, $\C^n\setminus 0$ is playing the role of a holomorphic analogue of~$\mathrm{Conf}_2(\R^{2n})$;  one physical interpretation of our higher symmetry algebras is that  nonlocal operators play an important role, giving rise to algebraic structures on  local operators via holomorphic descent. We expect that it is possible to study a holomorphic  analogue of the operad of little disks, and to use it to characterize secondary operations in holomorphic theories at the level of operads which can imposed concretely via a holomorphic analog of descent; however, we do not pursue this here, reserving such questions for future work.

Using the formalism of factorization algebras and the pairing discussed above, higher analogues of Kac--Moody algebras were recently introduced in~\cite{FHK,GWkm}. It was then argued in~\cite{GWkm,ISBW} that these algebras appear naturally in holomorphic twists of four-dimensional field theory as twists of the current supermultiplet associated to a global symmetry. 
A natural higher analogue of the Virasoro algebra in holomorphic theories was also proposed in~\cite{KH,Wthesis}; there is a model which uses the Dolbeault resolution of the Lie algebra of holomorphic vector fields on $\C^n\setminus 0$. (The reader will recall that the Virasoro algebra is a central extension of the Witt algebra of holomorphic vector fields on~$\C\setminus 0$.) 
Central extensions generalizing the well-known Kac--Moody and Virasoro central extensions were shown to exist. 
The space of {\em local} central extensions of the higher Virasoro algebra was shown in~\cite{Wthesis} to be two-dimensional; it is natural to guess that these cocycles correspond to the $a$ and~$c$ central charges of four-dimensional conformal field theory, and we hope to make this connection explicit in future work.

In the present work, our aim is to explore the relation of the four-dimensional higher Virasoro algebra to superconformal symmetry in the full theory. 
We compute the holomorphic twist of the four-dimensional superconformal algebra in~\S\ref{sec: twist}, and argue that the resulting algebra, $\lie{sl}(3|\N-1)$ (or $\lie{psl}(3|3)$ when $\N=4$), acts naturally as a finite-dimensional closed subalgebra of the holomorphic vector fields on an appropriate holomorphic superspace, $\C^{2|\N-1}$. 
See Theorem \ref{thm: twist} for the precise statement. 

In \S\ref{sec:hol} we define the main holomorphic symmetry algebras of interest which are natural from the point of view of complex geometry. 
In~\S\ref{sec:enhance}, we show that the holomorphic twists of supersymmetric theories in four dimensions admit natural actions of these holomorphic symmetry algebras (at the classical level). At this stage it plays no role if the untwisted theory is in fact superconformal or not. 
For the precise results pertaining to symmetry enchancements of twists of four-dimensional theories, see Propositions~\ref{prop: symenhance1}--\ref{prop: symenhance4}.

In~\S\ref{sec: fact}, we introduce factorization algebras associated to the various symmetry algebras and theories. 
\S\ref{sec: defsym} then goes on to consider further deformations of the enhanced symmetry algebra as a factorization algebra, which play a role in other twists of the theory. 
We consider, in particular, the deformation of the centrally extended higher Kac--Moody and Virasoro factorization algebras by a Maurer--Cartan element arising from a special supersymmetry in the global superconformal algebra, making connection with the work of Beem and collaborators~\cite{BeemEtAl}. 
For us, this deformation appears as a simple odd vector field implementing the Koszul resolution of a subspace $\C\subset\C^2$. 
The main results of \S\ref{sec: defsym} can be summarized as follows.

\begin{thm}\label{thm: intro}
Let $\cF_{\rm KM}$ and $\cF_{\rm Vir}$ be the $\cN=2$ higher Kac--Moody and Virasoro factorization algebras of level $\kappa$ and charge $c$, respectively, on $(z_1,z_2) \in \CC^2$.
(See Definitions~\ref{dfn: kmfact} and~\ref{dfn: virfact} respectively). 
Let $\cF_{\rm KM}'$ and $\cF_{\rm Vir}'$ be the corresponding factorization algebras deformed by the Maurer--Cartan element $z_2 \frac{\partial}{\partial \ep}$ arising as a special supercharge in the global superconformal algebra (see \S\ref{sec: defsym}).
Then:
\begin{itemize}
\item $\cF_{\rm KM}'$ is equivalent to a stratified factorization algebra on $\CC^2$ which is trivial away from $z_2 = 0$ and along the plane $\CC_{z_1} = \{z_2 = 0\}$ is equivalent to the Kac--Moody vertex algebra at level $- \kappa / 2$. 
\item $\cF_{\rm Vir}'$ is equivalent to the stratified factorization algebra on $\CC^2$ which is trivial away from $z_2 = 0$ and along the plane $\CC_{z_1} = \{z_2 = 0\}$ is equivalent to the Virasoro vertex algebra at level $- 12 c$.
\end{itemize} 
\end{thm} 
In words, at the level of stratified factorization algebras, the deformations of our higher symmetry algebras reproduce the chiral algebras studied in~\cite{BeemEtAl}; these are strictly contained within the full higher Virasoro and Kac--Moody symmetries, and are obtained from them by a further twist. We also reproduce the correct correspondence between four- and two-dimensional central extensions of these algebras; this is strong evidence that the central extensions of the higher algebra correspond precisely to the higher-dimensional central charges of the physical theory, just as in two dimensions.

Other examples of chiral algebras have been extracted from four-dimensional $\N=2$ theories, and we expect that the higher Virasoro and Kac--Moody algebras can profitably be used to understand all of them. Many of these appear from further twists; for example, the half-holomorphic twist of~\cite{Kapustin} is implemented by a natural further deformation. We expect that the recent results of~\cite{OhYagi}, producing chiral algebras isomorphic to those of~\cite{BeemEtAl} from this half-holomorphic twist in the presence of an $\Omega$-deformation, can be swiftly understood in our setting. The first study of infinite-dimensional symmetry at the level of the holomorphic twist was performed in~\cite{JohansenKM}, but was restricted to the setting of a product of Riemann surfaces; for us, the essential geometry for the study of local operators in four dimensions is that of $\C^2\setminus 0$. However, it is worth noting that the formalism of factorization algebras allows us to think of symmetries by local Lie algebras across all complex surfaces uniformly.

In addition, we emphasize that the symmetry enhancement in the holomorphic theory means that \emph{many more} deformations of the differential are available after the holomorphic twist. Of course, any appropriate Maurer--Cartan element of the global superconformal algebra gives rise to such a deformation, but new deformations appear in the holomorphic twist which are not visible at the level of the full theory. While we reserve a full analysis for future work, we explore some of these new deformations briefly in~\S\ref{ssec: moredefs}, arguing that the higher Virasoro algebra in $\N=2$ supersymmetric theories admits a deformation that localizes it to the holomorphic vector fields on \emph{any} smooth affine curve in~$\C^2$, not just to planes. We expect even more interesting behavior in the case of singular or nonreduced curves, though we do not explore this direction further here.

In~\S\ref{sec: qft}, we turn to some explicit examples of theories, and in particular to $\N=2$ super QCD. We demonstrate that the higher Virasoro symmetry is, in fact, anomalous, and can be realized in the quantum theory precisely when the familiar condition $N_f = 2 N_c$ is satisfied---i.e., when the full theory is in fact superconformal. The beta function of the full theory is thus visible as an anomaly in the holomorphic twist---in spite of the fact that the holomorphic theory itself is automatically scale-invariant. We then offer a precise characterization of the chiral algebras (or two-dimensional holomorphic theories) that appear upon deforming $\N=2$ superconformal QCD as above, again reproducing results of~\cite{BeemEtAl} in our formalism.




\subsection*{Acknowledgements}
We thank K.~Costello, S.~Gukov, O.~Gwilliam, Si Li, S.~Nawata, K.~Nilles, D.~Pei, P.~Yoo and J.~Walcher for conversations and advice that helped lead to the completion of this work. Special thanks are due to R.~Eager for early discussions and calculations, related to those in~\S3, that helped spark this project. I.S. thanks the Center for Quantum Geometry of Moduli Spaces for hospitality during the preparation of this work. B.W. thanks Ruprecht-Karls-Universit\"at Heidelberg and the Max-Planck-Institut f\"ur Mathematik for hospitality during the preparation of this work. The work of I.S. was supported in part by the Deutsche Forschungsgemeinschaft, within the framework of the Exzellenzstrategie des Bundes und der L\"{a}nder. The work of B.W. was supported by Northeastern University, the University of Edinburgh, and National Science Foundation Award DMS-1645877.

\section{Superconformal algebras and their twists}
\label{sec: twist}

\subsection{The conformal and superconformal symmetry algebras}

We here review some basic notions of conformal and superconformal symmetry in physical theories. Our index conventions are standard; indices for the vector representation of an orthogonal group are raised and lowered with the metric. We sometimes use the isomorphism between the vector representation of $\lie{so}(4)$ and the tensor product of its two chiral spinors; spinor indices are raised and lowered with the $\lie{su}(2)$-invariant alternating form $\epsilon_{\alpha\beta}$. 

\begin{dfn}
  The conformal algebra in dimension~$d>2$, with signature $(p,q)$, is $\lie{so}(p+1,q+1)$.
\end{dfn}

\begin{prop}[Standard; see for example~\cite{CFTbook}]
  \label{prop: VFs}
  The conformal algebra acts by vector fields on~$\R^{p,q}$.
\end{prop}
\begin{proof}
  This is essentially by definition, since the conformal group is the set of diffeomorphisms of~$\R^{p,q}$ that act by local rescaling on the metric. We remind the reader that the relevant vector fields form a finite-dimensional algebra in dimensions greater than two, and can be explicitly given as
  \deq[eq:confVF]{
    \begin{aligned}[c]
      P_\mu &= \pdv{ }{x^\mu}, \\
      M_{\mu\nu} &= x_\mu \pdv{ }{x^\nu} - x_\nu \pdv{ }{x^\mu}, \\
      \Delta &= - E, \\
      K_\mu &= |x|^2 \pdv{ }{x^\mu} - 2 x_\mu E.
    \end{aligned}
  }
  Here $E = x^\mu \partial_\mu$ is the Euler vector field. It is straightforward to check that these satisfy the commutation relations
  \deq{
    \begin{aligned}[c]
      [D,P_\mu] &= P_\mu, \\
      [D,K_\mu] &= - K_\mu, \\
      [M_{\mu\nu}, K_{\rho}] &= g_{\rho\nu} K_{\mu} - g_{\mu\rho} K_\nu, \\
      [M_{\mu\nu}, P_{\rho}] &= g_{\rho\nu} P_{\mu} - g_{\mu\rho} P_\nu, \\
      [K_\mu, P_\nu] &= 2 M_{\mu\nu} + 2 g_{\mu\nu} E,\\
      [M_{\mu\nu}, M_{\rho\sigma}] &= g_{\mu\sigma} M_{\nu\rho} + g_{\nu\rho} M_{\mu\sigma} - g_{\mu\rho} M_{\nu\sigma} - g_{\nu\sigma} M_{\mu\rho},
    \end{aligned}
    }
    with other commutators vanishing. For a proof that these vector fields span the space of solutions to the conformal Killing vector field equations, see~\cite{CFTbook}.
\end{proof}

\begin{rmk}
  \label{rmk: flags}
  In dimension four, the accidental isomorphism $\lie{so}(6)\cong\lie{su}(4)$ gives rise to a convenient way of thinking about the vector fields defined above. Let us pass to using complex coefficients. We can realize the action of the complexified conformal group on complexified Minkowski space $\C^4$ by considering the quotient of~$GL(4,\C)$ by a particular parabolic subgroup:
  \deq{
    {\rm Fl}(2;4) = GL(4,\C)/P, \quad P = \begin{bmatrix} * & * \\ 0 & * \end{bmatrix}.
  }
  (We choose to use $GL(4,\C)$, rather than $SL$, for the sake of convenience; note, however, that the unit-determinant condition can be imposed everywhere, and does not affect our discussion.) Here $P$ consists of two-by-two blocks with the lower left block zero and other blocks arbitrary. The resulting symmetric space is the space of 2-flags in~$\C^4$; it has an open dense subset isomorphic to~$\C^4$, given by cosets represented by matrices of the form
  \deq[eq:blockX]{
    \begin{bmatrix} 1 & 0 \\ x_{\alpha\dot\alpha} & 1 \end{bmatrix} \in GL(4,\C).
  }
  The reader will recall that the chiral (or antichiral) spinor of~$\Spin(6)$, equivalent to the fundamental (or antifundamental) representation of~$SU(4)$, becomes one chiral and one anti-chiral spinor of~$\Spin(4)\cong SU(2)\times SU(2)$, which can be thought of as sitting block-diagonally inside~$SU(4)$. Our index convention in~\eqref{eq:blockX} is meant to suggest this. The vector fields witnessing the natural action of~$GL(4,\C)$ on~$\Fl(2;4)$ from the left become the conformal vector fields of~\eqref{eq:confVF} when restricted to the image of this embedding of Minkowski space in~$\Fl(2;4)$.
\end{rmk}

The construction of Remark~\ref{rmk: flags} becomes even more important in the context of superconformal symmetry.
In  a limited number of cases classified by Nahm~\cite{Nahm}---in particular, when the spacetime dimension does not exceed six---the conformal algebra can be extended to supersymmetric theories, which then admit the action of a simple \emph{superconformal algebra} $\lie{c}(d,\N)$ containing both the conformal algebra and the $\N$-extended super Poincar\'e algebra. A complete list of such algebras for dimension greater than two is
\deq{
  \lie{sconf}(d,\N) = 
  \begin{cases}
    \lie{osp}(\N|4), & d=3; \\
    \lie{su}(2,2|\N), & d = 4, ~  \N\neq 4; \\
    \lie{psu}(2,2|4), & d = 4, ~ \N = 4; \\
    \lie{f}(4), & d=5, ~  \N = 1; \\
    \lie{osp}(6,2|\N), & d = 6.
  \end{cases}
}
In each case, the construction relies on an accidental isomorphism of Lie algebras, akin to that used in Remark~\ref{rmk: flags}, that allows one to fit the spinor representations of low-dimensional spin groups into Kac's classification of simple super Lie algebras~\cite{Kac}, where no infinite families with odd elements in spinor representations appear.

We now specialize to four-dimensional theories, and thus to the algebras~$\lie{su}(2,2|\N)$ for $\N=1$ and~$2$, and $\lie{psu}(2,2|4)$ in the case $\N=4$. In our considerations, we will always complexify, and thus deal with the complex Lie algebras $\lie{sl}(4|\N)$ or $\lie{psl}(4|4)$. (The change for $\N=4$ comes about because $\lie{sl}(k|k)$ has a one-dimensional center and is therefore not simple; algebras with $\N>4$ exist, but are not of physical relevance, as they cannot be represented on interacting theories.) One can helpfully think of the generators of this algebra as  arranged in the following diagram: 
  \begin{equation}
    \begin{tikzcd}[row sep = 1 ex, column sep = 1 ex]
      & & P_{\alpha\dot\alpha} & &\\
      & Q_\alpha^i & & \bar{Q}_{\dot\alpha i} & \\
      M_{\alpha\beta} & & \Delta, R^i_j & & \bar{M}_{\dot\alpha \dot\beta} \\
      & S_{i}^\alpha & & \bar{S}^{\dot\alpha i} & \\
    & & K^{\alpha \dot\alpha } & &
    \end{tikzcd}
    \label{eq:SCAdiagram}
  \end{equation}
  Here, vertical position in the diagram represents the conformal weight of the corresponding generator, and horizontal position is determined by the difference in number of chiral and antichiral spinor indices. If vertical position is interpreted as a $\Z$-grading, parity is determined by its value modulo two. The charges $Q_\alpha^i$ and~$\bar S^{\dot\alpha i}$ together form a chiral  spinor of~$\cO(4,2)$, which  is  equivalent to  the antifundamental representation of~$\su(2,2)$;  they transform in the fundamental representation of the $R$-symmetry group. Likewise, $\bar{Q}_{\dot\alpha i}$ and~$S_i^\alpha$ together sit in the fundamental representation of~$\lie{su}(2,2)$ and the antifundamental representation of~$R$-symmetry. (We generally follow the conventions of~\cite{DolanOsborn}.)

The superconformal algebra acts naturally by vector fields on supermanifolds. For example, the usual superspace for $\N=1$ supersymmetry in four dimensions is $\R^{4|4}$, with one odd copy of each chiral spinor representation; it admits an action of $\lie{su}(2,2|1)$ by supervector fields that extends the natural action of super Poincar\'e by supertranslations. The vector fields were written explicitly in~\cite{VongDuc77}, and shown to arise (as in the bosonic case) from conjugating super Poincar\'e transformations by the superspace analogue of the inversion transformation. This generalizes to extended superconformal symmetry; see~\cite{HoweHartwell} for details. Here, a consistent real structure can be imposed, such that the odd part of~$\R^{4|4}$ is the Majorana spinor of~$SO(3,1)$. However, this will play no role in our further considerations.
  
  However, the standard (unconstrained) superspace is not the only superspace where the superconformal algebra naturally acts. Of particular interest for us will be an action on superfields satisfying a chiral constraint.
\begin{prop}
  The complexified four-dimensional superconformal algebra $\lie{sl}(4|\N)$ acts geometrically by supervector fields on the chiral superspace
  \deq{
    \C^{4|2\N} \cong V \oplus \Pi(S_+ \otimes R),
  }
  where $V$ and~$S_+$ denote respectively the fundamental and chiral spinor representations of~$\Spin(4)$, and $R$ the defining representation of the $U(\N)$ $R$-symmetry.
\end{prop}
\begin{proof} As in Proposition~\ref{prop: VFs}, this is almost a proof by definition, although the characterizations involved seem to be less well-known in this case. 
  One can in fact define the four-dimensional superconformal algebra to be the collection of vector fields on unconstrained superspace, $\R^{4|4\N}$, that (after complexification) act compatibly with every possible chiral constraint. That is,
  \deq{
    [X,D_{\alpha i}] \sim D_{\alpha i}, \quad
    [X,\bar D_{\dot\alpha}^j] \sim D_{\dot\alpha}^j.
  }
  For general~$\N$, this characterization is given, for instance, in~\cite{HoweWest};
  see~\cite{Osborn} for further discussion and an explicit treatment of the case~$\N=1$. 

  Since the superconformal transformations act preserving chiral subspaces, they also act on each chiral subspace. In the case $\N=1$, the explicit supervector fields involved are
  \deq{
    P_{\alpha\dot\alpha} = \partial_{\alpha\dot\alpha}, \quad M_{\alpha\beta} = x_{\alpha\dot\alpha} \partial_{\beta\dot\alpha} + \theta_\alpha \partial_\beta + (\alpha \leftrightarrow \beta),
    \quad \bar{M}_{\dot\alpha\dot\beta} = x_{\alpha\dot\alpha} \partial_{\alpha\dot\beta} + (\dot\alpha \leftrightarrow \dot\beta)
  }
  for generators of affine transformations, as well as
  \deq{
    \Delta = E + \frac{1}{2} \theta^\alpha \partial_\alpha,
    \quad
    R = \theta^\alpha\partial_\alpha, 
    \quad
    K^{\alpha\dot\alpha} = |x|^2 \partial^{\alpha\dot\alpha} - 2 x^{\alpha\dot\alpha} E +  \theta^\alpha x^{\beta\dot\alpha} \partial_\beta
  }
  for dilatations and $U(1)$ $R$-symmetry, and special conformal transformations, and lastly
  \deq{
    \begin{aligned}[c]
    Q_\alpha &= \partial_\alpha, \\
    \bar{Q}_{\dot\alpha} &= \theta^\alpha \partial_{\alpha\dot\alpha},
  \end{aligned}
  \qquad
  \begin{aligned}[c]
    S^\alpha &= - x^{\alpha\dot\alpha} \theta^\beta \partial_{\beta\dot\alpha} + \theta^2 \partial^\alpha,\\
    \bar{S}^{\dot\alpha} &= x^{\alpha\dot\alpha} \partial_\alpha
  \end{aligned}
  }
  for the fermionic transformations. Here $E$ again denotes the Euler vector field. We further  note that when $\N=4$, the action factors through the simple quotient $\lie{psl}(4|4)$. 
\end{proof}
\begin{rmk}
  Note that the action of the conformal algebra is modified from its purely bosonic form! While the supervector fields realizing supertranslations remain unaffected, $K$ now contains fermion-dependent terms. However, under the quotient map from functions on superspace to functions on (bosonic) $\C^4$, the vector fields of Proposition~\ref{prop: VFs} are reproduced.  
\end{rmk}

As above in Remark~\ref{rmk: flags}, it is extremely helpful to justify the existence of such vector fields by relating the affine superspace that carries this group action to a symmetric space constructed directly from the superconformal group. In doing this, we follow the excellent discussion in~\cite{HoweHartwell}; the interested reader is referred there for more information. 
\begin{dfn}
  Let $\C^{m|n}$ be a supervector space. A \emph{flag} is a sequence of subobjects in the category of supervector spaces,
  \deq{
    0 \subset V_1 \subset \cdots \subset V_k \subset \C^{m|n}, \quad \dim V_i = m_i|n_i,
  }
  where each containment is strict. As in the usual case, flags are characterized by their \emph{type}, which is the list of dimensions $\{m_i|n_i\}$. These must form a strictly increasing sequence inside of the  poset $\Z_+ \times \Z_+$. We will denote the space of flags of  a  fixed type by~$\Fl(m_i|n_i;m|n)$.
\end{dfn}

\begin{obs}
  The flag manifold $\Fl(m_i|n_i;m|n)$ naturally carries an action of~$GL(m|n,\C)$, exhibiting it as a symmetric space. 
  As in the usual case, we can see this by exhibiting the space of flags as the right quotient of~$GL(m|n,\C)$ by the stabilizer of a standard flag of appropriate type. We form a standard flag of type $m'|n'$ by fixing an ordered basis of~$\C^{m|n}$, considered as a $GL(m|n,\C)$ module in the standard way, and taking the flag spanned by the first $m'$ even and the first $n'$ odd basis vectors. The left $GL(m|n,\C)$ action on the flag variety remains unbroken and gives rise to a subalgebra of the vector fields on~$\Fl(m_i|n_i;m|n)$ representing~$\lie{gl}(m|n)$. For example, the parabolic subgroup stabilizing a standard flag of type $m'|n'$ consists of block matrices of the form
  \deq{
    \left[ \begin{array}{cc|cc}
        * & * & * & * \\
        0 &  * &  0  & * \\ \hline
        *  & * & *  & *  \\
        0 & * & 0  &  *
      \end{array}
    \right]
    \left[ \begin{array}{c}
        * \\ 0 \\ \hline * \\ 0 
      \end{array}
    \right]
  }
  as a subgroup of~$GL(m|n,\C)$.
\end{obs}

We are now equipped to give supersymmetric analogues of the construction  of the conformal compactification $\Fl(2;4)$ of four-dimensional affine space in Remark~\ref{rmk: flags}. 
\begin{prop}[\cite{HoweHartwell}]
  The left-chiral $\N$-extended superspace in four dimensions is a dense open subset in~$\Fl(2|0;4|\N)$. Similarly, the right-chiral superspace is a dense open subset in~$\Fl(2|\N;4|\N)$, and the full superspace $\C^{4|4\N}$ admits a compactification to~$\Fl(2|0,2|\N;4|\N)$. 
\end{prop}
We note that the map for the chiral superspaces can be represented by matrices of the form
\deq{
  \left[
    \begin{array}{cc|c}
      1 & 0 & 0 \\
      x^{\alpha\dot\alpha} & 1 & 0 \\ \hline
    \theta^{\alpha}_{i} & 0 & 0 
    \end{array}
  \right],
  \qquad
  \left[
    \begin{array}{cc|c}
      1 & 0 & 0 \\
      x^{\alpha\dot\alpha} & 1 & \bar\theta^{\dot\alpha i} \\ \hline
      0 & 0 & 0 
    \end{array}
  \right]
}
respectively.
The reader is referred to~\cite{HoweHartwell} for the proof and further discussion.
\begin{rmk}
  For future use, it is  helpful to summarize the correspondence for the reader between superconformal generators as tabulated above and a matrix presentation of~$\lie{sl}(4|\N)$ with the following diagram:
  \deq[eq:confmats]{
    \left[
      \begin{array}{cc|c}
        M^\alpha_\beta & K^{\alpha\dot\beta} & S^\alpha_j \\
        P_{\alpha\dot\beta} & \bar{M}_{\dot\alpha}^{\dot\beta} & \bar{Q}_{\dot\alpha j} \\ \hline
        Q_\alpha^i & \bar{S}^{\dot\alpha i} & R^i_j
      \end{array}
    \right]
    ; \qquad
    \Delta = \left[
      \begin{array}{cc|c}
        1 & 0  & 0 \\
        0 & -1 & 0 \\ \hline
        0 & 0 & 0
    \end{array}\right],
    \qquad
    r = \frac{1}{\N-4}\left[
      \begin{array}{cc|c}
        \N & 0 & 0 \\
        0 & \N & 0 \\ \hline
        0 & 0 & 4 
    \end{array}\right].
  }
\end{rmk}

\subsection{Twisted superconformal symmetry}
We now proceed to consider consequences of superconformal symmetry for holomorphically twisted theories. It is sensible to begin by deforming the superconformal algebra to a dg-Lie algebra, using the differential generated by the adjoint action of the holomorphic supercharge, and then by computing its cohomology. The resulting algebra will act on any holomorphically twisted superconformal theory.

\begin{thm}
  \label{thm: twist}
Let $\sca(\N)$ be the complexified superconformal algebra in four dimensions. The cohomology of~$\sca(\N)$ with respect to a holomorphic supercharge is $\sca^\hol(\N) = \lie{sl}(3|\N-1)$ for $\N=1$ and~$2$, or $\lie{psl}(3|3)$ for $\N=4$. 
\end{thm}
  We proceed by computing the cohomology directly; it is a quotient of the commutant of~$Q$. To begin, we change from Lorentz indices to holomorphic notation, adapted to the symmetry group left unbroken in $Q$-cohomology; upon breaking $SO(4)$ to~$U(2)$, the left-chiral spinor index becomes a pair of charged singlets, labeled by $\pm$, and the right-chiral spinor becomes the fundamental of the unbroken $SU(2)$. As for $R$-symmetry indices, we break $U(\N)$ to $U(1) \times U(\N-1)$; label the corresponding indices $0$ and~$i$, with position of the index recording fundamental versus antifundamental  representations of  the corresponding groups.
  
  After having done this, we  can represent a basis for the algebra using a diagram  analogous  to~\eqref{eq:SCAdiagram} above:
    \begin{equation}
    \begin{tikzcd}[row sep = 1 ex, column sep = 1 ex]
      & & P^+_{\dot\alpha}, P^-_{\dot\alpha} & &\\
    & Q_+^i, Q_-^i, Q_+^0, Q_-^0  & & \bar{Q}_{\dot\alpha i}, \bar{Q}_{\dot\alpha 0} & \\
    M_{++}, M_{--}, M_{+-} & & \Delta, R^i_j, R^0_j, R^i_0, R^0_0 & & \bar{M}_{\dot\alpha \dot\beta} \\
  & S^{+}_{0}, S^{-}_{0}, S^{+}_i, S^{-}_i & & \bar{S}^{\dot\alpha i}, \bar{S}^{\dot\alpha 0} & \\
    & & K^{\dot\alpha }_+, K^{\dot\alpha}_- & &
    \end{tikzcd}
  \end{equation}
  We choose the holomorphic supercharge to be $Q_+^0$. 
  One can then simply use the commutation relations given above to determine exact pairs; the conclusion is that the holomorphic momenta $P^-_{\dot\alpha}$ and superconformal transformations $K^{\dot\alpha}_+$ survive, together with $\bar{M}_{\dot\alpha\dot\beta}$ and the traceless $R$-symmetry $R_i^j$. The surviving fermions are $\bar{Q}_{\dot\alpha}^i$ and~$Q^i_-$ from the Poincar\'e supercharges, and $\bar{S}^{\dot\alpha i}$ and~$S^-_i$ from the conformal supersymmetries. Together with two additional central generators, these implement the algebra $\lie{sl}(3|\N-1)$.

  We have included this discussion to orient those readers with a physics background, using relatively standard notation, as to which of the conformal symmetries survive in the holomorphically twisted theory. However, a proof of the proposition that is both less cumbersome and more useful can be given by just considering the matrix group $GL(4|\N)$ together with a parabolic subgroup, and computing the holomorphic twist directly. The advantage is that one directly obtains a description  of the twist of the module $GL(4|\N)/P$, with its action of the twisted superconformal algebra.  (As above,  it is convenient to ignore the traceless condition for the moment and restore it later on.)

\begin{proof}[Proof of Theorem~\ref{thm: twist}]
  Identifying the conformal algebra with~$\lie{sl}(4|\N)$ as in~\eqref{eq:confmats} above, the holomorphic supercharge corresponds to the elementary matrix generator $e_{0+}$. (We continue to use the index set~$\alpha,\dot\alpha,i$ for a basis of the supervector space $\C^{4|\N}$; $+$ and~$0$, as in the previous discussion, denote specific values of these indices, and we will use $\mu$ for any element of this basis, chosen without specifying parity or spin.) Using the standard commutation relations, we see that
  \deq{
    [e_{0+}, e_{\mu\nu}] = \delta_{+\mu} e_{0\nu} - (-)^{|\mu\nu|} \delta_{\nu 0} e_{\mu +},
  }
  which immediately implies that 
  $\ker(\ad_Q)$ is spanned by elementary matrices with $\mu\neq +$ and $\nu \neq 0$, together with $e_{00} + e_{++}$, and that
  $\im(\ad_Q)$ is spanned by the elementary matrices $e_{0\nu}$ and~$e_{\mu +}$ (allowing only the diagonal combination $e_{00}+e_{++}$).
  The cohomology is therefore isomorphic to~$\lie{sl}(3|\N-1)$; if we follow the parabolic subgroup defining the chiral superspace through the same computation, we find matrices of the form
  \deq{
    \left[\begin{array}{ccc|c}
        - & - & - & * \\
        0 & - & - & * \\
        0 & - & - & * \\ \hline
        0 & * & * & *
      \end{array}
    \right]
  }
  Looking just at the bosonic part of this calculation (or, equivalently, setting $\N=1$), the reader will recognize the parabolic subgroup defining $\Fl(1;3)\cong\C P^2$ as a maximally symmetric space for the group~$SL(3)$. In general, the resulting coset space is~$\Fl(1|0;3|\N-1)$. Very similarly to the untwisted case, an open dense subset is the holomorphic affine superspace $\C^{2|\N-1} = \Spec \C[z_1,z_2;\epar_i]$, where $\epar_i$ ($1\leq i < \N$) are fermionic scalars. 
\end{proof}

\begin{cor}
  The twisted superconformal algebra $\lie{sl}(3|\N-1)$ acts geometrically by holomorphic supervector fields on~$\C^{2|\N-1}$.
\end{cor}
Indeed, it is easy to describe these vector fields explicitly. 
In the case $\N=1$, no fermions remain, so that we just need to give an action of~$\lie{sl}(3)$ by holomorphic vector fields on~$\C^2$. A straightforward calculation shows that the vector fields
  \deq{
    p_i = \pdv{}{z_i}, \qquad
      m_{ij} = z_i \pdv{ }{z_j} , \quad
      k_i = z_i e 
  }
  give the desired module structure.
  Here $e = z_i \partial/\partial{z_i} = \tr(m)$ is the holomorphic Euler vector field.

  In the general case, we need to add additional even vector fields to implement the $R$-symmetry, as well as fermionic vector fields in the appropriate representations of~$\lie{sl}(3)\times\lie{sl}(\N-1)$. We must also modify the vector field implementing the conformal weight to 
  \deq{
    z_i \pdv{ }{z_i} + \frac{1}{2} \epar \pdv{ }{\epar},
  }
  although this of course just amounts to a change of basis in the Cartan subalgebra. Further, we must replace the Euler vector field in the definition of the generators $k_i$ by
  $\hat e = z_i \partial/\partial{z_i} + \epar \partial/\partial\epar$
  (in the case $\N=2$). The needed odd vector fields in this case are
  \deq{
    \pdv{ }{\epar}, \quad \epar\pdv{ }{z_i}
  }
  for positive conformal weight, and 
  \deq{
    z_i\pdv{ }{\epar}, \quad \epar e
  }
  for negative conformal weight. In general, we obtain one copy of this for every odd parameter; the $R$-symmetry is of course implemented by the vector fields $\epar_i \partial/\partial \epar_j$.


\section{Derived structures in complex geometry}
\label{sec:hol}

We will demonstrate that there are natural enhancements of the twisted superconformal algebras computed above to certain infinite-dimensional Lie algebras. 
The Lie algebras will be defined as the {\em derived} global holomorphic sections of certain (graded) holomorphic vector bundles which enlarge the holomorphic tangent bundle. 
We will use a convenient model for derived sections given by the Dolbeault resolution of a holomorphic vector bundle.
Throughout, the reader should bear in mind the familiar process in two dimensions by which the holomorphic M\"obius transformations are enhanced to the Witt algebra of holomorphic vector fields on~$\C^\times$ and subsequently centrally extended in the quantum theory to the Virasoro algebra; central extensions of higher symmetry algebras will be discussed below in~\S\ref{sec: fact}.  

In parallel, we introduce a similar enhancement of global symmetry in holomorphically twisted theories to a local Lie algebra, analogous to Kac--Moody symmetry in two dimensions. 
Additionally, we review results on holomorphic twists of supersymmetric field theories in four dimensions. 


\subsection{Local Lie algebras and symmetries}
\label{ssec: localLie}

A local Lie algebra is a graded vector bundle $L^\bu$ on $X$ equipped with differential and bi-differential operators which turn the corresponding sheaf of sections $\cL^\bu$ into a sheaf of dg Lie algebras. \footnote{There is also a version of this for $L_\infty$-algebras, in which the structure maps are required to be poly-differential operators.}
For a precise definition we refer the reader to~\cite[Definition 6.2.1]{CG1}. 

Throughout this section we fix a smooth complex surface $X$. 
There are two varieties of local Lie algebras on $X$ that will be of interest to us: (1) Lie algebras of holomorphic currents which arise as resolutions of the sheaf of holomorphic $\fg$-valued functions on $X$, and (2) Lie algebras of holomorphic vector fields on $X$. 

\subsubsection{Lie algebras of holomorphic currents}
\label{sssec: localcur}

These local Lie algebras are the natural enhancements of global (flavor) symmetries in holomorphically twisted theories.
\begin{dfn}
Let $\fg$ be a Lie algebra.
The local Lie algebra of {\em $\cN$-extended holomorphic $\fg$-currents} on a complex surface $X$ is
\begin{equation}
  \cG_{\cN}^\bu (X) = \Omega^{0,\bu}(X, \fg \otimes_{\CC} A),
\end{equation}
where $A = \O(\C^{\N-1}[1]) = \CC[\varepsilon_1,\ldots,\varepsilon_{\N-1}]$, and $\varepsilon_i$ are variables of cohomological degree $-1$. 
\end{dfn}

The algebras just mentioned live naturally over a complex manifold $X$ of any dimension. 
Indeed, when $\cN=1$, it is simply given as the Dolbeault complex on $X$ with values in $\fg$.
For extended supersymmetry, we can also give a geometric interpretation that thinks of them as objects living over a certain graded space.

Let $X$ be a complex manifold of dimension $d$.
For any $m \geq 0$, define the graded space $X^{d|m}$ to have graded ring of functions 
\[
  \cO(X^{d|m}) = \cO(X) [\varepsilon_1,\ldots, \varepsilon_{m}] = \O(X) \otimes \O(\C^m[1]),
\]
where $\varepsilon_i$ are variables of cohomological degree $-1$.
Note that we here treat the odd directions as completely algebraic, and will persist in this convention. 
Thus, for instance, when we write $\Omega^{p,q}(X^{d|m})$ we mean forms of type $(p,q)$ on $X$ with values in the graded ring $\CC[\varepsilon_1,\ldots, \varepsilon_m]$. 

Another way of describing this operation is to say that we are forming the trivial holomorphic bundle with fiber $\C^{m}$ over~$X$, and then defining $X^{d|m}$ to be its parity shift. Of course, there are many other supermanifolds with body~$X$---we could, for example, consider the parity shift of an arbitrary holomorphic bundle---but this family are appropriate for our present purposes. This is indicated by the fact that, after the holomorphic twist of the module of chiral superfields in~\S\ref{sec: twist}, all remaining fermions transformed as scalars under the structure group.

We can thus interpret the $\N$-extended holomorphic currents as just consisting of the Dolbeault complex with coefficients in~$\lie{g}$, but taken on the $\N$-extended space $X^{2|\N-1}$:
\deq{
  \cG_\cN^\bu (X) = \Omega^{0,\bu}(X,\lie{g}) \otimes_\C \O(\C^{\N-1}[1]) = \Omega^{0,\bu}(X^{2|\N-1}, \lie{g}).
}

\subsubsection{Holomorphic vector fields}
\label{sssec: localVF}

Let $T_X$ be the holomorphic tangent bundle.
There is a natural resolution of this sheaf by vector bundles given by the Dolbeault complex
\[
\fX_1^{\bu}(X) = \Omega^{0,\bu}(X, T_X) 
\]
which is equipped with its natural $\dbar$ operator. 
(The subscript $1$ is to be consistent with notation that we will introduce momentarily.)
On a $\dbar$-acyclic open set, this resolution is quasi-isomorphic to holomorphic vector fields. 
The Lie bracket of holomorphic vector fields extends naturally to $\fX_1^\bu(X)$ to give it the structure of a sheaf of dg Lie algebras. 

The differential and bracket on $\fX_1^{\bu}(X)$ are given by differential and bidifferential operators, respectively. 
Thus, $\fX_1^\bu(X)$ defines a local Lie algebra on $X$. 
As a sheaf of dg Lie algebras, $\fX_1^\bu(X)$ is equivalent to the sheaf of holomorphic vector fields. 
However, the sheaf of holomorphic vector fields is {\em not} a local Lie algebra since it is obviously not given as the $C^\infty$-sections of a vector bundle.
We will refer to $\fX_1^\bu(X)$ as the {\em local Lie algebra of holomorphic vector fields} (and will omit the bullet for cohomological degree unless necessary). 

Let us introduce an extended version of this dg Lie algebra.
For $A$ a graded commutative algebra, we denote its graded Lie algebra of derivations by ${\rm Der}(A)$. 

\begin{dfn}\label{dfn: vf}
The local Lie algebra of {\em $\cN$-extended holomorphic vector fields} on a complex surface $X$ is
\[
\fX_{\cN}^\bu (X) = \left( \Omega^{0,\bu}(X, T X) \otimes_\CC A \right) \bowtie \left( \Omega^{0,\bu}(X) \otimes_\CC {\rm Der}(A) \right) 
\]
where $A = \O(\C^{\N-1}[1])$ as above. 
(Note that $\N=1$ extended holomorphic vector fields are just holomorphic vector fields again, since no fermions survive the twist of the $\N=1$ algebra.) 

The symbol $\bowtie$ here denotes a direct sum of dg vector spaces, but equipped with a different Lie algebra structure.
The desired dg Lie structure can be described concretely as follows:
\begin{itemize}
\item the differential is $\dbar$ on both summands in the above decomposition;
\item the Lie bracket on $\Omega^{0,\bu}(X, T X) \otimes_\CC A$ is obtained from tensoring the ordinary Lie bracket on vector fields with the graded commutative product on $A$. 
That is, if $X \otimes a$ and $X' \otimes a'$ are sections, then the bracket is
\[
[X \otimes a, X' \otimes a'] = [X,Y] \otimes a a' ;
\]
\item the Lie bracket on $\Omega^{0,\bu}(X) \otimes_\CC {\rm Der}(A))$ is obtained from tensoring the graded commutative wedge product on differential forms with the Lie bracket on derivations on $A$.
That is, if $\omega \otimes D$ and $\omega' \otimes D'$ are sections, then the bracket is
\[
[\omega \otimes D , \omega' \otimes D'] = (\omega \wedge \omega') \otimes [D,D'] ;
\]
\item the remaining brackets are through the Lie derivative of holomorphic vector fields on $X$ and the natural action of~$\Der(A)$ on~$A$.
\end{itemize}
\end{dfn}

Just as in the case of the current algebras associated to a Lie algebra, there is an interpretation of these local Lie algebras of vector fields as vector fields living on a certain graded manifolds. 
If $X$ is a complex manifold and $\cN \geq 1$, we have the graded manifold $X^{d|\cN-1}$. 
Its holomorphic tangent bundle $T X^{d|\cN-1}$ has as its space of sections $\Gamma(X^{d|\cN-1}, T X^{d|\cN-1})$ which splits as a vector space $\Gamma^{hol}(X, TX)[\ep_1,\ldots, \ep_{\cN-1}] \oplus \cO^{hol}(X) \otimes \Der(\CC[\ep_1,\ldots, \ep_{\cN-1}]$).
The local Lie algebra $\fX_{\cN}$ is a resolution of this sheaf of holomorphic section, where we only resolve by forms on the manifold, and treat the odd directions as algebraic.  


\begin{notation}
When $X = \CC^2$ we will abbreviate the local Lie algebras $\cG_{\cN}(\CC^2)$ and $\fX_{\cN}(\CC^2)$ by $\cG_{\cN}$ and $\fX_{\cN}$ respectively. 
\end{notation}

%

\subsection{Holomorphic theories on complex surfaces}
\label{ssec: theories}

In this section we introduce some classes of holomorphic field theories defined on complex surfaces.
We work inside of the BV formalism so that the space of fields is equipped with a $(-1)$-shifted symplectic pairing, see \cite{CosRenorm, CG2, CostelloHol} for the requisite background. 
We recall how these theories arise as holomorphic twists of $\cN=1,2,4$ supersymmetric Yang--Mills theory in four dimensions in the next section.

We start with the simplest holomorphic gauge theory, which is a holomorphic analog of a familiar topological theory. 

\begin{dfn}
Let $\fh$ be a $\ZZ$-graded Lie algebra \footnote{A similar definition applies for any $L_\infty$ algebra.} and $X$ a complex surface.
{\em Holomorphic BF theory} on $X$ with values in $\fh$ is the BV theory whose fields are
\begin{align*}
\ul{A} & \; \in \Omega^{0,\bu}(X, \fh)[1] \\
\ul{B} & \; \in \Omega^{2,\bu}(X, \fh^\vee) .
\end{align*}
with action functional 
\[
S(\ul{A}, \ul{B}) = \int_X \<\ul{B} \wedge F_{\ul{A}}\>_{\fh} = \int_X \<\ul{B} \wedge \dbar \ul{A}\>_{\fh} + \frac{1}{2} \int_X \<\ul{B} \wedge [\ul{A}, \ul{A}]\>_{\fh}
\]
where $\<-,-\>_{\fh}$ denotes the graded symmetric pairing between $\fh$ and $\fh^\vee$. 
\end{dfn}

If $\fh$ is equipped with a graded {\em skew}-symmetric invariant pairing $\<-\>$ and $X$ is equipped with a holomorphic volume form $\Omega$ , then there is a different action we can write down
\[
S (\ul{A}) = \int_X \Omega \wedge \<\ul{A} \wedge F_{\ul{A}}\> = \frac12 \int_X \<\ul{A} \wedge \dbar \ul{A}\> + \frac{1}{6} \int_X \<\ul{A} \wedge [\ul{A}, \ul{A}]\>
\]
which only depends on the field $A$. 
We refer to this as {\em holomorphic Chern--Simons theory} on $X$ with values in $\fh$. 

\subsubsection{}

We now turn to field theories describing holomorphic analogs of matter and linear $\sigma$-models.

\begin{dfn}
Let $\VV$ be a finite dimensional graded vector space and $L$ a line bundle on a complex surface $X$. 
The {\em holomorphic $\beta\gamma$ system on $X$}, twisted by $L$, with values in $\VV$, is the the BV theory whose fields are
\begin{align*}
\ul{\gamma} \in \; & \Omega^{0, \bu}(X, L) \otimes \VV \\
\ul{\beta} \in \; & \Omega^{2, \bu}(X, L^\vee) \otimes \VV^\vee [1] 
\end{align*}
with action functional
\[
S(\ul{\beta},\ul{\gamma}) = \int_{X} \<\ul{\beta} \wedge \dbar \ul{\gamma}\>_{L \otimes \VV} .
\]
Here, the braces $\<-,-\>_{L \otimes \VV}$ denotes the graded symmetric pairing between sections of $L \otimes \VV$ and $L^\vee \otimes \VV^\vee = (L \otimes \VV)^\vee$.
\end{dfn}

The graded vector space $\VV$ may not be concentrated in a single degree, as this example indicates. 

\begin{eg}
A typical example concerns the graded vector space $\VV = V[\varepsilon] = V[1] \oplus V$, where $V$ is an ordinary vector space and $\varepsilon$ is a formal parameter of degree $-1$. 
In this case, we can use the Berezin integral to identify
\[
\VV^\vee \cong V^\vee[\varepsilon][-1] .
\]
The pairing between $\VV$ and $\VV^\vee$ is 
\[
(v + \varepsilon v' , \phi + \varepsilon \phi') \mapsto \int_{\CC^{0|1}} \<v + \varepsilon v' , \phi + \varepsilon \phi'\>_V = \<v, \phi'\> + \<v', \phi\>
\]
where $\<-,-\>_V$ is the dual pairing between the ordinary vector spaces $V$ and $V^\vee$. 

Thus, for this particular $\VV = V[\varepsilon]$ we can think of the $\beta\gamma$ system twisted by $L$ as a theory on the graded manifold $X^{2|1} = T[-1] X$, where the fields are
\begin{align*}
\ul{\gamma} \in \; & \Omega^{0, \bu}(X^{2|1}, L) \otimes V \\
\ul{\beta} \in \; & \Omega^{2, \bu}(X^{2|1}, L^\vee) \otimes V^\vee 
\end{align*}
and the action is
\[
S(\ul{\beta},\ul{\gamma}) = \int_{X^{2|1}} \<\ul{\beta} \wedge \dbar \ul{\gamma}\> .
\]
\end{eg}

\begin{eg}
We can further simplify a special case of this theory when we have made an additional choice on the complex surface $X$. 
Suppose we choose a square root of the canonical bundle on $X$.
Then, the $\beta\gamma$ system, twisted by $L = K_X^{\frac12}$, with values in $\VV = V[\varepsilon]$ is equivalent to the theory with a single set of fields given by
\[
\ul{\varphi} \in \Omega^{0,\bu}(X, K_X^{\frac12}) \otimes T^*V \otimes \CC[\varepsilon] = \Omega^{0,\bu}(X^{2|1}, K_X^{\frac12}) \otimes T^*V
\]
where the action is
\[
S(\ul{\varphi}) = \int_{X^{2|1}} \<\ul{\varphi} \wedge \dbar \ul{\varphi}\> .
\]
\end{eg}

This example leads us to the following special case of a higher $\beta\gamma$ system.

\begin{dfn}\label{dfn: boson}
Let $(U, \omega)$ be a symplectic vector space and $K_X^{\frac12}$ a choice of a square root of the canonical bundle on the complex surface $X$.
The {\em holomorphic symplectic boson} system on $X$ with values in $Z$ is the BV theory whose fields are
\[
\ul{\varphi} \in \Omega^{0,\bu}(X^{2|1}, K_X^{\frac12}) \otimes U 
\]
which we write in components as $\ul{\varphi} = \varphi + \varepsilon \varphi' \in \Omega^{0,\bu}(X, K_X^{\frac12}) \otimes U [\varepsilon]$.
The action is
\[
S(\ul{\varphi}) = \frac{1}{2} \int_{X^{2|1}} \omega(\ul{\varphi} \wedge \dbar \ul{\varphi}) = \int_X \omega(\varphi \wedge \dbar \varphi') .
\]
\end{dfn}

\begin{rmk}
More generally, one can consider a $\sigma$-model of the form 
\[
X \to T[-1] U
\]
where $(U,\omega)$ is an arbitrary holomorphic symplectic {\em manifold}. 
After twisting by $K_X^{\frac12}$, the AKSZ construction endows the (derived) space of maps ${\rm Map} (X, T[-1] U)$ form with a $(-1)$-shifted symplectic structure. 
\end{rmk}

To write the theory in the notation of the $\beta\gamma$ system, we simply take the symplectic vector space $Z = T^*V$. 

\subsection{Holomorphic Noether currents} 
\label{sec:noether}

In the BV formalism the space of local functionals $\oloc$ is equipped with the BRST operator which can be written as $\{S,-\}$ with $S$ the classical BV action. 
Furthermore, the BV bracket $\{-,-\}$ encodes the cohomological shift by one $\oloc [-1]$ of the space of local functionals with the structure of a dg Lie algebra. 
The first cohomology $H^1$ of this dg Lie algebra encodes infinitesimal deformations of the theory.
The zeroth cohomology $H^0$ encodes infinitesimal automorphism. 
A symmetry by a Lie algebra $\fg$ is given by a map of Lie algebras $I \colon \fg \to \oloc [-1]$. 
For $X \in \fg$ the local functional $I_X$ should be thought of as a local version of the Noether current associated to the symmetry $X$. 
For more on this perspective we refer to \cite[Chapter 3]{CG2}. 

Consider holomorphic BF theory for the graded Lie algebra $\fg [\ep_1,\ldots, \ep_{\cN-1}]$. 
This theory has a symmetry by the Lie algebra $\fX_{\cN}$ of holomorphic vector fields on $X^{2|\cN-1}$ defined by the following Noether current 
\beqn\label{eqn:noethervf}
I (\ul{\xi}) = \int_X \<\ul{B} \wedge L_{\ul{\xi}} \ul{A}\> .
\eeqn
Where $L_{\ul{\xi}}$ denotes the Lie derivative by the graded vector field $\ul{\xi} \in \fX_{\cN}$. 

It is immediate to check that $\ul{\xi} \mapsto I(\ul{\xi})$ defines a map of Lie algebras from $\fX_\cN$ to the space of local functionals. 
This is equivalent to the following master equation
\[
\d_{CE} I + \{S, I\} + \frac12 \{I, I\} = 0 
\]
where $\d_{CE}$ is the Chevalley--Eilenberg differential for holomorphic vector fields, $S$ is the classical action of BF theory, and $\{-,-\}$ denotes the BV bracket. 

Abstractly, the existence of this symmetry is manifest. 
Consider the case $\cN=1$, for simplicity. 
Holomorphic BF theory describes the formal moduli space of holomorphic $G$-bundles on $X$ near the trivial $G$-bundle. 
Deformations of the $\dbar$ operator on the trivial $G$-bundle are of the form $\dbar + A$ where $A$ is some $\Omega^{0,1}(X)$ form satisfying the Maurer--Cartan equation $\dbar A + \frac12 [A,A] = 0$. 
We can also consider deformations of the complex structure which modifies the $\dbar$-operator to $\dbar + \mu$ with $\mu \in \Omega^{0,1} (X, T_X)$ satisfying the Maurer--Cartan equation $\dbar \mu + \frac12 [\mu, \mu] = 0$. 
The space of infinitesimal automorphisms of such a deformation are holomorphic vector fields on $X$. 

For the $\beta\gamma$ system valued in $\VV = V[\ep_1,\ldots, \ep_n]$ there is a completely similar formula for the symmetry by $\cN$-extended holomorphic vector fields.
Next, suppose that $V$ is a $\fg$-representation. 
Then, the $\cN$-extended holomorphic current algebra $\cG_\cN$ is also a symmetry with local Noether currents defined by 
\beqn\label{eqn:noetherg}
I (\ul{\alpha}) = \int \<\ul{\beta} \wedge (\ul{\alpha} \cdot \ul{\gamma})\> .
\eeqn
There is a completely similar formula for a current algebra symmetry on theory of the holomorphic symplectic boson. 

\section{Derived symmetry enhancement}
\label{sec:enhance}

In this section we reach the main conceptual leap of this work.
We show that upon performing the holomorphic twist of a four-dimensional supersymmetric gauge theory there is a significant enhancement of symmetries. 
We focus on two classes of symmetries present in the (untwisted) supersymmetric theory: global (flavor) symmetries and superconformal symmetries. 
Global symmetries are described by a finite dimensional Lie group.
Further, in \S \ref{sec: twist} we have recalled the finite dimensional algebras describing superconformal symmetries. 
In both cases we find an enhancement of symmetries to infinite dimensional Lie algebras which are of a similar spirit to affine and Virasoro algebras in chiral CFT.

\subsection{A catalog of results about twisting}
\label{ssec: twisted theories}

We summarize the characterization of the holomorphic twists of four-dimensional supersymmetric Yang--Mills theories, see \cite{CostelloHol, ESW} for a formulation of these results in a manner which is closest to our setup.

Recall, the (complexified) supertranslation algebra in four dimensions is the $\ZZ/2$-graded Lie algebra 
\[
\ft_{\cN} = \CC^4 \oplus \Pi(S_+ \otimes \CC^{\cN} \oplus S_- \otimes \CC^\cN)
\]
where $\CC^4$ is the complexified abelian Lie algebra of translations, and $S_\pm$ are the positive/negative spin representations of $\fs \fo(4)$. 
There is a nontrivial Lie bracket determined by Clifford multiplication
\[
\Gamma : S_+ \otimes S_- \to \CC^4 .
\]
For more details on supersymmetry algebras, we recommend~\cite{NV} or~\cite[\S3.1]{Chris}. 

By a holomorphic supercharge, we mean an odd square-zero element of the supertranslation algebra $Q \in \ft_{\cN}$ such that the image of $[Q,-]$ (which lies in $\CC^4$) is of rank two.  
To this data, one defines the holomorphically twisted theory as in~\cite[\S 15]{CostelloHol}
It was observed in~\cite{CostelloHol, Chris} that such a supercharge always exists in four dimensions, and any two choices of a holomorphic supercharge gives rise to equivalent theories up to conjugation.

We summarize the results of twisting with respect a holomorphic supercharge, starting with $\cN=1$ supersymmetry. 

%

\begin{prop}[Well-known; for various treatments, see~\cite{ESW, Johansen, Yangian, CostelloHol, ISBW}] \label{prop: 1twist}The holomorphic twist of $\cN=1$ super Yang--Mills with values in an ordinary Lie algebra $\fg$ coupled to the chiral multiplet with values in a representation $V$ is equivalent to the coupled holomorphic $BF$-$\beta\gamma$ system where $\fh = \fg$ and $\VV = V$. 
\end{prop}

Next, we move on to $\cN=2$ supersymmetry. 

\begin{prop}[\cite{CostelloHol,ESW}] \label{prop: 2twist}
The holomorphic twist of $\cN = 2$ super Yang--Mills with values in an ordinary Lie algebra $\fg$ coupled to the hypermultiplet with values in a symplectic representation $V$ is equivalent to holomorphic $BF$ theory with values in $\fh = \fg [\varepsilon]$ coupled to the holomorphic symplectic boson with values in $Z$. 
\end{prop}

\begin{rmk}
On affine space $X = \CC^2$ the canonical bundle is trivial so the theory of the holomorphic symplectic boson with values in the symplectic vector space $V$ is equivalent to the $\beta\gamma$ system with values in $W$ where $W$ is any subspace satisfying $V = T^*W$. 
Thus, the holomorphic twist of the $\cN=2$ hypermultiplet on $\CC^2$ is equivalent to a particular $\beta\gamma$ system on $\CC^2$.
\end{rmk}

\begin{rmk}
This is a general remark about a convention that we are taking for the holomorphic twist of $\cN=2$. 
As explained in~\cite{CostelloHol,ESW}, part of the data one needs to twist a field theory is that of a {\em twisting homomorphism}.
This is a group homomorphism
\[
\rho : \CC^\times \to G_{R}
\]
where $G_R$ is the $R$-symmetry group, with the property that the weight of the twisting supercharge $Q$ under $\rho$ is $+1$. 
For $\cN=2$, the $R$-symmetry group is $\GL_2(\CC)$, so there are different choices for $\rho$ one can make given a fixed supercharge. 
Recall, the odd part of the supertranslation algebra for $\cN=2$ is of the form
\[
S_+ \otimes \CC^2 \oplus S_- \otimes \CC^2
\]
where $S_{\pm}$ are the positive and negative irreducible spin representations of $\fs \fo(4, \CC)$. 
The holomorphic twist corresponds to choosing a $Q$ of the form
\[
Q = q \otimes \begin{bmatrix} 1 \\ 0 \end{bmatrix} \in S_+ \otimes \CC^2 .
\]
Up to conjugation, there are two choices for $\rho$ for which such a $Q$ has weight $+1$. 
They are
\[
\rho_1(t) = \begin{bmatrix} t & 0 \\ 0 & t \end{bmatrix} \; \;\; {\rm and} \; \;\; \rho_2(t) = \begin{bmatrix} t & 0 \\ 0 & t^{-1} \end{bmatrix}  
\]
Both $\rho_1$ and $\rho_2$ lead to holomorphic theories, but they differ in their respective presentations as a BV theory. 

One can show that $\rho_1$ leads to the description of twisted $\cN=2$ supersymmetry that we present here. 
The choice of $\rho_2$ leads to a very similar holomorphic theory, with the only difference that the cohomological degree of $\ep$ is $+1$, instead of the $-1$ that we use.
\end{rmk}

Finally, we finish with the result of the holomorphic twist of maximal supersymmetry. 

\begin{prop}[\cite{BaulieuCS, CostelloHol, CY4, ESW}] \label{prop: 4twist}
The holomorphic twist of $\cN=4$ super Yang--Mills with values in an ordinary Lie algebra $\fg$ is equivalent to holomorphic BF theory with values in $\fg [\ep_1,\ep_2]$. 
When $\fg$ is equipped with a non-degenerate symmetric invariant pairing (for instance, if $\fg$ is semi-simple) it admits an equivalent description as holomorphic Chern--Simons theory with values in~$\fg[\ep_1,\ep_2,\ep_3]$.
\end{prop}

The isomorphism between holomorphic Chern-Simons for the graded Lie algebra $\fg[\ep_1,\ep_2,\ep_3]$ and holomorphic BF theory for the graded Lie algebra $\fg[\ep_1,\ep_2]$ can be seen as follows. 
Holomorphic Chern-Simons on $\CC^{2|3}$ has fields
\[
\ul{A} \in \Omega^{0,\bu}(\CC^2, \fg[\ep_1,\ep_2,\ep_3]) .
\]
Recall, in order to construct holomorphic Chern--Simons theory we need a graded skew-symmetric invariant pairing on the Lie algebra and a holomorphic volume form on the complex surface.
For the graded Lie algebra $\fg[\ep_1,\ep_2,\ep_3]$ the invariant pairing is defined by
\[
(\ul{X}_1 , \ul{X}_2) \mapsto \int_{\CC^{0|3}} \d^3 \ep  \, \<\ul{X}_1 , \ul{X}_2\> 
\]
for $\ul{X}_1,\ul{X}_2 \in \fg[\ep_1,\ep_2,\ep_3]$, where $\<-\>$ is the {\em symmetric} pairing on $\fg$. 
On the complex surface $\CC^2$ we choose the holomorphic volume form $\d^2 z = \d z_1 \d z_2$. 
With these choices, the action of holomorphic Chern--Simons theory on the odd Calabi--Yau space $\CC^{2|3}$ can be written as
\[
\int_{\C^{2|3}} \d^2 z \, \d^3 \ep \, \left(\<\ul{A} \wedge \dbar \ul{A}\> + \frac{1}{6} \<(\ul{A} \wedge [\ul{A} , \ul{A}]) \>\right) .
\]

The invariant pairing identifies $\fg$ with $\fg^*$. 
Thus, we can write the fields of holomorphic BF theory as
\begin{align*}
(\ul{A}, \ul{B}) & \in \Omega^{0,\bu}(\CC^2, \fg[\ep_1,\ep_2])[1] \oplus \Omega^{2,\bu} (\CC^2, \fg [\ep_1, \ep_2]) \\ \ul{A} + \delta \ul{B}  & \cong \Omega^{0,\bu}(\CC^2, \fg[\ep_1,\ep_2])[1] [\delta]
\end{align*}
where $\delta$ is a parameter of degree one. 
The correspondence between fields of holomorphic Chern-Simons and BF theory can be realized by $\ep_1 \leftrightarrow \ep_1$, $\ep_2 \leftrightarrow \ep_2$, and $\ep_3 \leftrightarrow \delta$. 

\subsection{Holomorphic symmetry enhancement}

There are two type of symmetries of supersymmetric theories that we focus on.
The first is a global (or flavor) symmetry by a Lie algebra $\fg$. 
For instance, any supersymmetric theory of matter in some $\fg$-representation has such a symmetry. 
The other is superconformal symmetry, which makes sense in $\cN=1,2$ or $4$ supersymmetry. 
In this section we see how the twists of the supersymmetric theories we have just cataloged have enhanced symmetries by enlargements of the (twists) of a global $\fg$ symmetries and a superconformal symmetry.

For instance, if a classical supersymmetric theory has a classical global symmetry by a Lie algebra $\fg$, then the holomorphically twisted theory has a symmetry by the local Lie algebra $\cG_{\cN}$.
Likewise, the superconformal algebra gets enchanced to a symmetry by the Lie algebra of (graded) holomorphic vector fields $\fX_{\cN}$. 

The precise statement for $\cN=1$ is the following. 

\begin{prop}\label{prop: symenhance1}
Suppose we start with a theory on $\RR^4$ with $\cN=1$ supersymmetry and a classical global symmetry by a Lie algebra $\fg$ which commutes with the supersymmetry algebra. 
Then, for any holomorphic supercharge $Q$, the twisted theory has a classical symmetry by the following local Lie algebras:
\begin{itemize}
\item holomorphic $\fg$-currents: $\cG_{1} = \Omega^{0,\bu}(\CC^2, \fg)$;
\item holomorphic vector fields: $\fX_{1} = \Omega^{0,\bu}(\CC^2, T \CC^2)$.
\end{itemize}
\end{prop}
\begin{proof}
By Proposition \ref{prop: 1twist}, the twist of a general $\cN=1$ theory is equal to holomorphic BF theory coupled to a holomorphic $\beta\gamma$ system. 
Since the global $\fg$ symmetry commutes with $Q$, it follows that $\fg$ is a symmetry of the twisted theory. 
In particular, the action of $\fg$ commutes with $\dbar$ and hence extends to an action by the local Lie algebra $\Omega^{0,\bu}(\CC^2, \fg)$ in such a way that the original global symmetry by the Lie algebra $\fg$ is compatible with the embedding $\fg \hookrightarrow \Omega^{0,\bu}(\CC^2, \fg)$ by the constant functions. 
We wrote the explicit local Noether current for $\cG_1$ in Equation \eqref{eqn:noetherg}, in the case $\cN=1$. 

For the second part, we recall that local Lie algebra of holomorphic vector fields $\fX_{1}$ acts on the fields of the $BF - \beta\gamma$ system by Lie derivative. 
The local Noether current for $\fX_1$ was constructed in Equation \eqref{eqn:noethervf}, in the case $\cN=1$.
\end{proof}

There is an anomaly to quantizing the $\cG_1$ symmetry, see \cite{GWkm}, and an anomaly to quantizing the $\fX_1$ symmetry, see \cite{Wthesis}. 

The statement for $\cN=2$ is similar. 

\begin{prop} \label{prop: symenhance2}
Suppose is a theory on $\RR^4$ with $\cN=2$ supersymmetry and a classical global symmetry by a Lie algebra $\fg$ which commutes with the supersymmetry algebra. 
Then, for any holomorphic supercharge $Q$, the twisted theory has a classical symmetry by the following local Lie algebras:
\begin{itemize}
\item holomorphic $\fg$-currents on $\CC^{2|1}$: $\cG_{2} = \Omega^{0,\bu}(\CC^2, \fg[\ep])$;
\item holomorphic vector fields on $\CC^{2|1}$: $\fX_{2}$. 
\end{itemize}
\end{prop}

\begin{proof}
By Proposition \ref{prop: 2twist} the holomorphic twist is equivalent to holomorphic BF theory for the Lie algebra $\fg[\ep]$ coupled to the holomorphic symplectic boson valued in a symplectic vector space $(U,\omega_U)$. 
We wrote the explicit local Noether current for $\cG_2$ in Equation \eqref{eqn:noetherg}, in the case $\cN=2$ (this used a description in terms of the $\beta\gamma$ system).  

For the second part, we observe that holomorphic vector fields $\fX_2$ act on the fields of the $BF$ theory and the symplectic boson by Lie derivative. 
The local Noether current for $\fX_1$ was constructed in Equation \eqref{eqn:noethervf}, in the case $\cN=2$.
\end{proof}

There is no anomaly to quantizing the symmetry by $\cG_2$. 
We will see that there is an anomaly to quantizing the symmetry by $\fX_2$ in Section \ref{sec: qft}. 
We leave a full characterization of the anomaly to future work.

The case $\cN=4$ is slightly more subtle due to the two equivalent descriptions we just discussed.
Firstly, we ignore any global (flavor) symmetries and focus just on symmetries by graded holomorphic vector fields. 
On one hand, the holomorphic twist of $\cN=4$ can be described by holomorphic BF theory for the Lie algebra $\fg[\ep_1,\ep_2]$. 
On the other hand, after choosing a holomorphic volume form on $\CC^2$ it admits a description as holomorphic Chern--Simons theory for the Lie algebra $\fg[\ep_1,\ep_2,\ep_3]$. 

The presentation in terms of holomorphic BF theory endows the theory with a similar symmetry as in the $\cN=1,2$ cases above. 
The holomorphic twist of $\cN=1,2,4$ pure gauge theory is BF theory valued in a Lie algebra with one, two, or three parameters which admits a symmetry by $\fX_1,\fX_2,\fX_3$ respectively. 
On the other hand, holomorphic Chern--Simons theory has a {\em larger} symmetry algebra. 

Consider $\CC^{2|3}$ equipped with its odd holomorphic volume form $\d^2 z \d^3 \ep$.
The action of holomorphic Chern--Simons theory is only invariant under graded vector fields which preserve this odd holomorphic volume form:
\[
\bigg\{\ul{\xi} \in {\rm Vect}^{\rm hol} (\CC^{2|3}) \; | \; L_{\ul{\xi}} (\d^2 z \d^3 \ep) = 0 \bigg\} .
\]
Just as with all sheaves of holomorphic vector fields, this is {\em not} a local Lie algebra.
To present it as a local Lie algebra we resolve the conditions that the vector be holomorphic and (super) divergence-free. 
We have constructed $\fX_4$ as the local Lie algebra which resolves the sheaf of holomorphic vector fields on $\CC^{2|3}$. 

\begin{dfn}
Let $\fX^{div}_4$ be the following local Lie algebra on $\CC^2$.
As a bundle of cochain complexes it is
\[
\begin{tikzcd}
\ul{0} & \ul{1} \\
\fX_4 \ar[r, "\ul{\partial}"] & \Omega^{0,\bu}(\CC^{2|3}) 
\end{tikzcd}
\]
where 
\[
\ul{\partial} = \sum_{i=1,2} \frac{\partial}{\partial z_i} \frac{\partial}{\partial (\partial_{z_i})} - \sum_{a=1,2,3} \frac{\partial}{\partial \ep_a} \frac{\partial}{\partial (\partial_{\ep_a})} 
\]
is the super divergence operator. 
The Lie bracket extends the bracket on $\fX_4$ by declaring that graded vector fields act on $\Omega^{0,\bu}(\CC^{2|3})$ by Lie derivative. 
\end{dfn}

\begin{prop}\label{prop: symenhance4}
Consider $\N=4$ superymmetric Yang--Mills theory on~$\R^4$. 
For any holomorphic supercharge $Q$, the twisted theory has a classical symmetry by the local Lie algebra $\fX_{4}^{div}$ of super divergence-free holomorphic vector fields on~$\C^{2|3}$. 
\end{prop}
\begin{proof}
By Proposition \ref{prop: 4twist} the twist of $\N=4$ super Yang--Mills theory with gauge algebra $\fg$ is given by the holomorphic Chern--Simons theory whose fields are
  \deq{
    \ul{\ul{A}} \in \left(\Omega^{0,\bu}(\C^2,\fg)[\epar_1,\epar_2,\epar_3]\right) \cong \Omega^{0,\bu}(\C^{2|3},\fg),
  }
The action of $\fX^{div}_{4}$ is the obvious geometric one by graded Lie derivative.
The local Noether current for $\fX_4^{div}$ is 
\[
\frac12 \int_{\CC^{2|3}} \d^2 z \d^3 \ep \, \<\ul{A} L_{\ul{\xi}} \ul{A}\> 
\]
where $\ul{\xi} \in \fX_4$. 
\end{proof}

From this symmetry by $\fX_4^{div}$ we can restrict to a symmetry by the smaller local Lie algebra $\fX_3$ which most obviously acts on the description in terms of holomorphic BF theory. 
Indeed, there is an embedding of local Lie algebras
\[
\fX_3 \hookrightarrow \fX_4^{div}
\]
defined by 
\[
\ul{\zeta} \mapsto \ul{\zeta} + (\ul{\partial} \zeta) \ep_3 \frac{\partial}{\partial \ep_3} . 
\]
In this expression $\zeta$ is a graded vector field on $\CC^{2|2}$ and $\ul{\partial}$ denotes the graded divergence operator on this space. 
The graded vector field on the right hand side is automatically graded divergence-free.

\begin{rmk}
We expect the algebra $\fX_{4}^{div}$ to play a role for other holomorphic twists of theories with $\N=4$ supersymmetry. However, these all contain gravitational multiplets. We restrict our considerations in this work to theories with rigid supersymmetry, deferring consideration of holomorphically twisted supergravity to future work.
\end{rmk}

\section{Factorization algebras and holomorphic symmetries}
\label{sec: fact}

One of the central ideas of~\cite{CG1,CG2} is that the observables of a quantum field theory form a factorization algebra.
In essence, the structure maps of a factorization algebra generalize the notion of the `operator product expansion.' 
For instance, in complex dimension one, holomorphic factorization algebras recover the notion of a vertex algebra~\cite[\S 5]{CG1}.
Likewise, there is a precise sense in which the symmetry algebra of a theory also forms a factorization algebra.
So far, we have described symmetries using (sheaves of) Lie algebras.
The associated factorization algebra is called the {\em enveloping} factorization algebra; see below for a recollection.

There is a sense in which one can `twist,' or deform, an enveloping factorization algebra.\footnote{We are {\em not} referring to a twist by a supercharge.}
For quantum phenomena it is necessary to take these twists into account.
Like central extensions of Lie algebras,
{\em local} cocycles of a local Lie algebra parametrize such twists.
We will be most concerned with cocycles of degree $+1$, as these correspond to ordinary central extensions at the level of Lie algebras or vertex algebras that we get back to in the later sections.
We characterize certain local cocycles in the local Lie algebras $\cG_{\cN}$ and $\fX_{\cN}$ introduced in the last section.

\subsection{Enveloping factorization algebras}

The local cohomology of a local Lie algebra $L$ is version of Lie algebra cohomology where the cochains are required to satisfy a locality axiom.
This means that as a cochain defines on the sheaf of sections $\cL^{\otimes k} \to \CC$ must be given as the integral of a Lagrangian density involving differential operators applied to the sections of $\cL$. 
We denote by $\cloc^\bu(\cL)$ the local Chevalley--Eilenberg cochain complex which computes local Lie algebra cohomology. 
For a precise definition see \cite[\S 3.4]{CG1}.

The (twisted) enveloping factorization algebra is defined from the following two pieces of data:
\begin{itemize}
\item a local Lie algebra $L$, and
\item a local cocycle $\phi \in \cloc^\bu(\cL)$ of cohomological degree $+1$. 
\end{itemize}

The value of the enveloping  factorization algebra on an open set $U$ associated to this data is defined as a deformation of the Chevalley--Eilenberg cochain complex computing Lie algebra homology of $\cL_c(U)$:
\[
\clie_\bu (\cL_c(U)) = \left(\Sym \left(\cL_c(U)[1]\right), \d_{CE} + \phi \right) .
\]
We denote the enveloping factorization algebra by $\UU_\phi (\cL)$.
For more detailed definition we refer to \cite[\S 6.3]{CG1}.
This construction is simultaneously a generalization of the enveloping algebra of a Lie algebra and the {\em chiral} enveloping algebra of a Lie$^{\star}$ algebra as in~\cite{BD}.

There is a very natural reason for considering central extensions in the context of field theory.
Local Lie algebras, such as $\cG_{\cN}$ and $\fX_{\cN}$, exist as classical symmetries of a field theory, as we saw above in the twists of four-dimensional $\cN=1,2,4$. 
A natural question is whether or not these symmetries persist at the quantum level. 
In general, there are two possible scenarios. 
The first scenario occurs when there is an {\em internal} anomaly present in the symmetry. 
This can arise when the local Lie algebra acts on some interacting field theory (such as a gauge theory). 
In order for the symmetry to exist at the quantum level, it must be the case that all internal anomalies vanish. 
Second, even if internal anomalies vanish, it may be the case that the symmetry algebra only acts {\em projectively}.
This means that while the original algebra does not act at the quantum level, a central extension does. 

We remark that in field theory all anomalies and central extensions that we study are local.
So it is necessary to characterize the local cohomology of the local Lie algebras which act as symmetries. 

\subsection{Extensions in complex dimension one}
\label{ssec: 2d background}

As a warm-up, we review the types of central extension present for local Lie algebras on Riemann surfaces.
For any Riemann surface $\Sigma$, and Lie algebra $\fg$, we have the current algebra
\[
\cG_{2d} = \Omega^{0,\bu}(\Sigma, \fg) 
\]
as introduced in \S\ref{sec:hol}. 
Given an invariant pairing $\kappa \in \Sym^2(\fg^*)^\fg$ on obtains a local cocycle $\phi_1 (\kappa) \in \cloc^\bu (\cG_{2d})$ of degree $+1$ defined by
\[
\phi_{\rm KM}(\kappa) : (\alpha,\beta) \mapsto \frac{1}{2 \pi i} \int_{\Sigma} \kappa(\alpha \partial \beta) .
\]
It is shown in~\cite[\S 5.4, Theorem 5.4.2]{CG1} that the vertex algebra corresponding to the {\em twisted} factorization enveloping algebra $\UU_{\phi_{1}(\kappa)}(\cG_{2d})$ on $\Sigma = \CC$ is equivalent to the Kac--Moody vertex algebra at level $\kappa$. 
The global sections of the twisted factorization enveloping algebra over a general surface $\Sigma$ recovers the conformal blocks of the affine Kac--Moody algebra. 

For vector fields, one proceeds similarly. 
Look at the local Lie algebra
\[
\fX_{2d} = \Omega^{0,\bu}(\Sigma, T_\Sigma)
\]
where $T_\Sigma$ is the holomorphic tangent bundle. 
Up to scale, there is one nontrivial cocycle of degree $+1$:
\[
H^1_{\rm loc} (\fX_{2d}) \cong \CC
\]
generated by the cocycle $\psi_{\rm Vir} \in \cloc^\bu(\fX_{2d})$ defined by the formula
\deq[eq: 2d vir]{
\psi_{\rm Vir} : \left(\alpha(z) \frac{\partial}{\partial z}, \beta(z) \frac{\partial}{\partial z}\right) \mapsto \frac{1}{24} \frac{1}{2 \pi i} \int_{\Sigma} \partial_z \alpha(z) \partial (\partial_z \beta(z)) .
}
It is shown in~\cite{Wil-vir} that the vertex algebra corresponding to the twisted factorization enveloping algebra $\UU_{c \psi_{\rm Vir}} (\fX_{2d})$ is equivalent to the Virasoro vertex algebra of central charge $c$. 
The global sections of the twisted factorization enveloping algebra over a general surface $\Sigma$ recovers the conformal blocks of the Virasoro algebra.


\subsection{Extensions of Kac--Moody type}
\label{ssec: KMxtns}
We turn our attention to factorization algebras associated to the local dg Lie algebra $\cG_{\cN}$ on a complex surface $X$ introduced in \S \ref{sec:hol}. 

\subsubsection{The case $\cN=1$}
We have $\cG_1 = \Omega^{0,\bu}(X , \fg)$. 
Twisted enveloping algebras in this case have been studied in \cite{GWkm}.
It is shown that any invariant polynomial of degree three $\theta \in \Sym^{3} (\fg^*)^\fg$ gives rise the local cocycle 
\[
(\alpha_0, \alpha_1,\alpha_2) \mapsto \frac{1}{(2 \pi i)^2} \int_X \theta (\alpha_1 \partial \alpha_2 \partial \alpha_3) 
\]
on $\cG_1$.

\begin{dfn}\label{dfn: kmfact1}
The ($\cN=1$) {\em higher Kac--Moody factorization algebra} on $X$ is the twisted factorization enveloping algebra 
$\UU_{\theta} \left(\cG_{1}\right).$
\end{dfn}

In \cite{GWkm} it is shown that on $X = \CC^2 \setminus 0$ the twisted factorization enveloping algebra gives rise to the two-variable Kac--Moody algebra \cite{FHK}.

\subsubsection{The case $\cN=2$}

Let's move on to the case $\cN=2$, so we are looking at Dolbeault forms valued in the graded Lie algebra $\fg[\ep]$, where $\ep$ is of degree $-1$. 
As above we write $\alpha + \ep \alpha' \in \Omega^{0,\bu}(X, \fg[\ep])$ where $\alpha,\alpha'$ are Dolbeault forms with no $\ep$-dependence. 

\begin{lem}
Let $\omega \in \Omega^{1,hol}(X)$ be a $\partial$-closed holomorphic one-form. 
There are maps of cochain complexes
\beqn\label{eqn:kappa}
\begin{array}{ccccl}
\phi^{(2)}_{\omega} & : &  \Sym^2(\fg^*)^\fg [-1] & \to & \cloc^\bu\left(\cG_{2}\right) \\
& & \kappa & \mapsto & \displaystyle \left( (\alpha, \ep \alpha') \mapsto \frac{1}{(2 \pi i)^2} \int \kappa(\alpha \wedge \partial \alpha') \wedge \omega \right)
\end{array}
\eeqn
and
\beqn\label{eqn:theta}
\begin{array}{ccccl}
\phi^{(3)} & : &  \Sym^3(\fg^*)^\fg [-1] & \to & \cloc^\bu \left(\cG_{2}\right)  \\
& & \theta & \mapsto & \displaystyle \left( (\alpha_0, \alpha_1, \alpha_2) \mapsto \frac{1}{(2 \pi i)^2} \int \theta(\alpha_0 \wedge \partial \alpha_1 \wedge \partial \alpha_2) \right) .
\end{array}
\eeqn
\end{lem}

\begin{proof}
The result for $\phi^{(3)}$ follows from the result for  $\cN=1$ in \cite{GWkm}. 
So, all we need to check is that for each $\kappa \in \Sym^2(\fg^*)^\fg$ that $\d \phi^{(2)}_{\omega}(\kappa) = 0$ where $\d$ is the differential on the local Chevalley--Eilenberg complex. 
This differential splits into two parts $\d = \dbar + \d_{CE}$ where $\dbar$ is the usual $\dbar$-operator on $X$ extended to functionals in the natural way, and $\d_{CE}$ encodes the Lie bracket on $\fg[\ep]$. 
The term $\d_{CE} \phi^{(2)}_{\omega}(\kappa)$ vanishes since $\kappa$ is invariant. 
The term $\dbar \phi^{(2)}_{\omega}(\kappa)$ vanishes by the following:
\begin{align*}
(\dbar \phi^{(2)}_{\omega}(\kappa))(\alpha ,\ep \alpha') & =  \frac{1}{(2 \pi i)^2} \int_X \dbar \left(\kappa( \alpha \wedge \partial \alpha')\right) \wedge \omega \\
& =  \frac{1}{(2 \pi i)^2} \int_X \dbar \left(\kappa( \alpha \wedge \partial \alpha') \wedge \omega \right) \\
& =  \frac{1}{(2 \pi i)^2} \int_X \d_{dR} \left(\kappa( \alpha \wedge \partial \alpha') \wedge \omega \right) \\
& = 0 .
\end{align*}
In the second line we used the fact that $\omega$ is holomorphic. 
In the third line we have used the fact that $\partial \omega = 0$.
\end{proof}

\begin{rmk}
One can write these local cocycles as an integrals over superspace $\CC^{2|1}$.
For instance $\phi^{(2)}_{\omega}(\kappa)$ can be written as
\[
\frac{1}{(2 \pi i)^2} \int_{\CC^{2|1}} \kappa(\ul{\alpha} \wedge\partial \ul{\beta}) \wedge \omega.
\]
where $\ul{\alpha} = \alpha + \ep\alpha ' \in \cG_{2} = \Omega^{0,\bu}(\CC^{2|1}, \fg)$. 
\end{rmk}

We arrive at the following definition.

\begin{dfn}\label{dfn: kmfact}
Fix a $\partial$-closed holomorphic one-form $\omega$ on $X$ and invariant polynomials $\kappa, \theta$ on $\fg$ of degree $2,3$ respectively. 
The $\cN=2$ {\em higher Kac--Moody factorization algebra} on $X$ is the twisted factorization enveloping algebra 
\[
\UU_{\omega, \kappa, \theta} \left(\cG_{2}\right) .
\]
\end{dfn}

In Theorem \ref{thm: intro} we parametrized this factorization algebra on $\CC^2$ with respect to a single `level' $\kappa$.
In this notation this corresponds to taking $\omega = \d z_2$.

\subsection{Extensions of Virasoro type}
\label{ssec: VFxtns}

We now describe some local cocycles of the local Lie algebra of (possibly graded) holomorphic vector fields on a complex manifold.

\subsubsection{Holomorphic vector fields}

Recall that $\fX(X)$, introduced in Section \ref{sssec: localVF}, is the local dg Lie algebra given by the Dolbeault resolution of holomorphic vector fields on $X$.
The local cohomology was computed in~\cite{Wthesis}. 
An analogous result to this one is valid for holomorphic vector fields on a complex manifold of any complex dimension. 

\begin{thm}\cite[\S 4.5]{Wthesis}\label{thm:wthesis}
There is an isomorphism of graded vector spaces
\[
H^\bu_{\rm loc}(\fX(X)) \cong H^\bu_{\rm dR} (X) \otimes H_{GF}^\bu(\fw_2) [4] .
\]
Here $\fw_2$ is the Lie algebra of formal vector fields on the formal $2$-disk, and $H_{GF}^\bu(\fw_2)$ is its (reduced) Gelfand--Fuks cohomology. 
\end{thm}

In particular, this result implies that on $X = \CC^2$, the local cohomology of $\fX(\CC^2)$ is isomorphic to a shift of the Gelfand--Fuks cohomology of the Lie algebra of formal vector fields $\fw_2$. 

\begin{rmk}
For any graded vector bundle $E$ there is an embedding of local functionals inside of all functionals $\oloc(\sE) \hookrightarrow \cO_{red}(\sE)$.
This translates to an embedding of sheaves of cochain complexes $\cloc^*(\cL) \hookrightarrow \cred^*(\cL_c)$ for any local Lie algebra $\cL$. 
In the case of vector fields, there is a related cochain complex that has been studied extensively in the context of characteristic classes of foliations \cite{Fuks, Guillemin, LosikDiag, Bernstein}, and more recently in~\cite{KH}. 
Suppose, for simplicity, that $X$ is a compact manifold.
The (reduced) {\em diagonal cochain complex} is the subcomplex 
\[
{\rm C}^\bu_{\Delta,\rm red}(\fX(X)) \subset \cred^\bu(\fX(X))
\]
consisting of cochains $\varphi : \fX(X)^{\otimes k} \to \CC$ satisfying $\varphi(X_1,\ldots,X_k) = 0$ if $\bigcap_{i=1}^k {\rm Supp}(X_i) = \emptyset$. 
That is, the cocycle vanishes unless all of the supports of the inputs overlap nontrivially. 
The inclusion of the local cochain complex $\cloc^*(\fX(X)) \subset \cred^*(\fX(X))$ factors through this subcomplex to give a sequence of inclusions
\[
\cloc^\bu(\fX(X)) \hookrightarrow {\rm C}^\bu_{\Delta,\rm red}(\fX(X)) \hookrightarrow \cred^\bu(\fX(X)) .
\]
\end{rmk}

The theorem implies that local cohomology classes on any complex manifold are characterized by a pair of a de Rham cohomology class on $X$ together with a Gelfand--Fuks class. 
On $\CC^2$ there is an explicit formula for generating local cocycles of this cohomology. 
If $\xi$ is a holomorphic vector field on $\CC^2$, its Jacobian is the function valued $2 \times 2$ matrix whose $ij$ entry is $\partial_{z_i} \xi_j (z)$, where $\xi_i(z)$ is the $i$th component of the vector field $\xi$.
Similarly, if 
\[
\xi = \xi_1(z,\zbar) \frac{\partial}{\partial z_1} + \xi_2(z,\zbar) \frac{\partial}{\partial z_2} \in \fX(\CC^2)
\]
is a Dolbeault valued vector field, then its Jacobian $J \xi$ is the $2\times 2$ Dolbeault valued matrix whose $ij$ entry is
\[
L_{\partial_{z_i}} \xi_j \in \Omega^{0,\bu}(\CC^2) .
\] 

In degree one, the local cohomology of holomorphic vector fields on $\CC^2$ is two-dimensional
\[
H^1_{\rm loc}(\fX(\CC^2)) \cong H^3_{GF} (\fw_2) \cong \CC\<[K_1], [K_2]\> ,
\]
spanned by the cocycles $K_1, K_2$ which have the following explicit descriptions:
\begin{align*}
K_1 (\xi) & = \int_{\CC^2} \Tr(J \xi) \wedge \Tr(\partial J \xi) \wedge \Tr(\partial J \xi) \\
K_2 (\xi) & = \int_{\CC^2} \Tr(J \xi) \wedge \Tr(\partial J \xi \wedge \partial J \xi) - \int_{\CC^2} \Tr(J \xi \wedge J \xi) \wedge \Tr(\partial J \xi) .
\end{align*}

These two cocycles are the holomorphic analogs of the so-called $a$ and $c$ cocycles which describe conformal anomalies for theories on $\RR^4$~\cite{Wthesis}.

\subsubsection{Graded vector fields}
Next, we turn to local cohomology classes for the local Lie algebra $\fX (X^{2|\cN-1}) $, which in \S \ref{sec:hol} we understood as the local dg Lie algebra of holomorphic vector fields on the graded manifold $X^{2|\cN-1}$. 
The graded generalization of Theorem \ref{thm:wthesis} is:

\begin{thm}
\label{thm:gradedvf}
Let $X$ be a complex manifold of dimension $2$ and suppose $\cN \geq 1$ is an integer. 
There is an isomorphism of graded vector spaces
\[
H^\bu_{\rm loc}(\fX(X^{2|k})) \cong H^\bu_{\rm dR} (X) \otimes H_{GF,red}^\bu(\fw_{2|\cN-1}) [4] .
\]
where $\fw_{2|\cN-1}$ is the graded Lie algebra of formal vector fields on the formal graded $2|\cN-1$-disk
\[
\fw_{2|\cN-1} = {\rm Der}\left(\CC[[z_1,z_2, \ep_1,\ldots, \ep_{\cN-1}]]\right) .
\]
Here $\ep_i$ have cohomological degree $-1$. 
\end{thm}

We will not prove, or explicitly use, this result. 
We postpone its proof to future work.

Just as in the non-graded cases, the local cohomology of graded holomorphic vector fields is whittled down to an understanding of the Gelfand--Fuks cohomology of the corresponding graded Lie algebra of formal vector fields. 
\begin{itemize} 
\item We do not know of a computation of the full cohomology of $\fw_{2|1}$. 
The existence of the cocycle we introduce below in Definition \ref{def: N=2 VF cocycle} implies $\dim H_{GF}^5 (\fw_{2|1}) \geq 1$. 
The claim that it represents a nontrivial class follows from the fact that it localizes to the ordinary Virasoro cocycle, see Section \ref{sec: defsym}. 
\item When $\cN = 3$ a result of \cite{AFuchs} implies $H^\bu_{GF} (\fw_{2|2}) = H^{\bu - 4}_{dR} (\U(2))$. 
\item For $\cN \geq 4$ one has $H^\bu_{GF} (\fw_{2|\cN-1}) \cong H^{\bu}_{dR} (S^5)$ by \cite{AFuchs}. 
\item There is a completely similar version of this theorem in complex dimension one. 
It is shown in \cite{AFuchs} that $H^\bu(\fw_{1|\cN-1}) \cong H_{dR}^\bu(S^3)$ for all $\cN \geq 1$. 
This agrees with the well-known fact \cite{Kac} that up to scale there is a unique central extension of vector fields on the graded punctured disk giving rise to the super Virasoro Lie algebras.
\end{itemize}

We will only be concerned with the case $\cN=2$ from hereon.
In fact, there is just one class of cocycles we need to pay attention to. A complete characterization like in the case of $\cN=1$ will be the subject of future work. 

The definition is the following. 

\begin{dfn}
  \label{def: N=2 VF cocycle}
For $i=1,2$ define the local cocycle $\psi^{i} \in \cloc^\bu(\fX_{2})$ by the formula
\[
\psi_i \left(\xi ,\ep \xi' \right) = \frac{1}{(2\pi i)^2} \int \tr(J \xi) \wedge \partial \tr(J \xi ') \wedge \d z_i
\]
where $\xi, \xi' \in \Omega^{0,\bu}(\CC^2, T^{1,0})$ are Dolbeault valued vector fields on $\CC^2$. 
The cochain $\psi_i$ is independent of odd vector fields of the form $f(z_1,z_2, \ep) \frac{\partial}{\partial \ep}$. 
\end{dfn}

The verification that $\psi_i$ is a cocycle is a direct calculation similar to the Kac--Moody case above, and the details are left to the reader. 
We remark that as $\psi^i$ is a cocycle we can then form the twisted factorization enveloping algebra. 

\begin{dfn}\label{dfn: virfact}
The $\cN=2$ {\em higher Virasoro factorization algebra} on $\CC^2$ is the twisted factorization enveloping algebra 
\[
\UU_{c_1 \psi_1 + c_2 \psi_2} \left(\fX_{2}\right) .
\]
where $c_1,c_2 \in \CC$. 
\end{dfn}

In Theorem \ref{thm: intro} we parametrized this algebra with respect to a single `central charge' $c \in \CC$.
In this notation this is $c = c_1$.


\section{Superconformal localization and holomorphic factorization algebras}
\label{sec: defsym}

In this section we study instances of deformations of the higher dimensional holomorphic factorization algebras introduced in~\S\ref{sec:enhance}. 
These deformation arise from a class of Maurer--Cartan elements of $\fX_\cN$ ($\cN=1,2,4$) which, in turn, define deformations of the Lie algebra of holomorphic currents, holomorphic vector fields, and holomorphic field theories such as those arising from twists of supersymmetry. 
The algebra $\fX_\cN$ is the Dolbeault resolution of some graded extension of holomorphic vector fields; for any $\cN$ there is a familiar class of Maurer--Cartan elements describing deformations of complex structure on the underlying complex manifold $X$. 
We focus mostly on the case $\cN=2$ and deformations that are not of this type.\footnote{Tautologically, any Maurer--Cartan element describes a deformation of the complex structure on the {\em graded} complex manifold $X^{2|1}$, however.}




One such deformation, which in the untwisted superconformal algebra arises from a special supersymmetry, is given by the holomorphic vector field 
\deq[eq:supercharge]{
z_2 \frac{\partial}{\partial \ep} \in \fX_{2} (\CC^2) .
}
This is the supercharge considered by Beem et al.~\cite{BeemEtAl}; we will show in this section that the chiral algebras they consider agree precisely with the corresponding truncations of our higher symmetry algebras.

In Definition \ref{dfn: kmfact} we have defined a factorization algebra on $X = \CC^2$ for any choice of a closed holomorphic one-form on $\CC^2$ and invariant polynomials $\kappa, \theta$ of degree $2,3$ on $\fg$. 
We fix such polynomials once and for all and also
consider the holomorphic one-form $\d z_2$.
For any $k_{4d} \in \CC$ we have an associated twisted enveloping factorization algebra $\UU_{\d z_2,k_{4d}\kappa,\theta} (\cG_2)$ on~$\CC^2$. 
Here, $k_{4d}$ is the four-dimensional avatar of the level, it simply scales the fixed symmetric bilinear form $\kappa$. 

In \S \ref{sec: defcur}, we will see how the Maurer--Cartan element (\ref{eq:supercharge}) determines a deformation of this enveloping factorization algebra to a factorization algebra that we will denote by $\cF_{k_{4d}} (\fg)$. 

\begin{thm}
\label{thm:kmlocal}
The factorization algebra $\cF_{k_{4d}} (\fg)$ is trivial away from $\{z_2 = 0\} \hookrightarrow \CC^2$. 
The localized factorization algebra $\cF_{k_{4d}} (\fg) |_{\CC_{z_1}}$ is holomorphic and its associated vertex algebra is isomorphic to the Kac--Moody vertex algebra of level $-\frac{k_{4d}}{2}$. 
\end{thm}

We point out that the dependence on the cubic polynomial $\theta$ has completely disappeared upon deforming by \eqref{eq:supercharge}.
The point is that the local cocycle \eqref{eqn:theta} determined by $\theta$ is rendered homotopically trivial upon performing this deformation.

The result for the Virasoro algebra is similar in spirit. 
We consider the deformation $\cF_{c_{4d}}$ of the $\cN=2$ Virasoro factorization algebra on $\CC^2$ associated to the cocycle $c_{4d} \psi_2$ (see notation from Section \ref{ssec: VFxtns} by the Maurer--Cartan element $z_2 \partial_\ep$. 

\begin{thm}
\label{thm:virlocal}
The factorization algebra $\cF_{c_{4d}}$ is trivial away from $\{z_2 = 0\} \hookrightarrow \CC^2$. 
The localized factorization algebra $\cF_{c_{4d}} |_{\CC_{z_1}}$ is holomorphic and its associated vertex algebra is isomorphic to the Virasoro vertex algebra of central charge $-12 c_{4d}$.
\end{thm}

This localization phenomena can also be used to see the conformal blocks of the Kac--Moody and Virasoro vertex algebras. 
In terms of factorization algebras this amounts to studying the factorization algebras on the not on $\CC^2$ but on the complex surface $\Sigma_{z_1} \times \CC_{z_2}$ where $\Sigma_{z_1}$ is some arbitrary genus Riemann surface. 
The results we have stated go through without much more difficulty to show that these factorization algebras localize to the Kac--Moody and Virasoro factorization algebras on $\Sigma_{z_1}$. 
The global sections, or factorization homology, along $\Sigma_1$ recover the conformal blocks of the respective vertex algebras.

We remark that there are new classes of deformations available from the point of view of factorization algebras that did not exist in the symmetry algebras of untwisted $\cN=2$ supersymmetric field theories. 
For instance, as a generalization of the above example, for any holomorphic polynomial $f \in \CC[z_1,z_2]$, one can consider the graded vector field
\beqn
f(z_1,z_2)  \frac{\partial}{\partial \ep} \in \fX_{2} (\CC^2) .
\eeqn
This is a Maurer--Cartan element in the dg Lie algebra $\fX_{2}(\CC^2)$, and hence determines a deformation (at least at the classical level) of any holomorphic twist of a four-dimensional $\cN=2$ theory. While we do not  consider these new deformations explicitly in great detail here, we will offer some remarks on them in~\S\ref{ssec: moredefs} below.

\subsection{Localization for factorization algebras} 

Suppose $X \hookrightarrow Y$ is a closed embedding. 
A factorization algebra on a manifold $Y$ can be restricted to any open subset. 
We say a factorization algebra on $Y$ is {\em trivial} away from $X$ if $\cF|_{Y \setminus X}$ is equivalent to the trivial factorization algebra which takes value $\CC$ on every open subset.
One can think of the factorization algebra $\cF$ as being ``localized" along the submanifold $X$. 

We will define an explicit model for such a factorization algebra on $Y$ as a factorization algebra on the closed submanifold $X$. 
Of course, factorization algebras do not ``restrict" to closed submanifolds in a naive sense. 
To define the localized factorization algebra on the closed submanifold $X$ we will utilize a tubular neighborhood. 
For concreteness, we restrict to the case $Y = \RR^n$. 
Let $\pi : NX \to \RR^n$ be the normal bundle, so that there is a neighborhood of the zero section in $NX$ which is diffeomorphic to a tubular neighborhood ${\rm Tub}(X)$ of $X$.

\begin{dfn}
Suppose $X \hookrightarrow \RR^n$ is a closed submanifold with tubular neighborhood $\pi : {\rm Tub}(X) \to X$. 
The {\em restriction} of a (pre)factorization algebra $\cF$ on $\RR^n$ to $X$ is the (pre)factorization algebra 
\[
\cF|_X = \pi_* \left(\left. \cF \right|_{{\rm Tub}(X)} \right)
\]
on $X$. 
\end{dfn}

In the case that $\cF$ is trivial away from $X \subset \RR^n$, we say that $\cF$ {\em localizes} to the factorization algebra $\cF|_X$.

\begin{rmk}
For the class of factorization algebras we consider all constructions will be independent of the choice of a tubular neighborhood.
\end{rmk}

\subsection{A deformation of the current algebra} 
\label{sec: defcur}

We focus on the case $\cN=2$, but similar constructions work for any $\cN \geq 2$. 
We deform the holomorphic current algebra $\cG = \cG_{2} = \Omega^{0,\bu}(\CC^2, \CC[\ep])$ by modifying the differential using the the Maurer--Cartan element $z_2 \frac{\partial}{\partial \ep}$.
Define the local Lie algebra on $\CC^2$:
\deq{
\cG' = \left(\Omega^{0,\bu}(\CC^2, \fg[\varepsilon])\; , \; \dbar + z_2 \frac{\partial}{\partial \varepsilon} \right) .
}
Leaving the internal $\dbar$ differential implicit, we can view this deformation as a two-term complex
\beqn\label{twotermdol}
\begin{tikzcd}
\ul{-1} & \ul{0} \\
\varepsilon \; \Omega^{0,\bu}(\CC^2, \fg) \ar[r, "z_2 \frac{\partial}{\partial \varepsilon}"] & \Omega^{0,\bu}(\CC^2, \fg) 
\end{tikzcd}
\eeqn
The Lie bracket remains unmodified, identical to that on~$\cG_{\fg[\ep]}$. 
This deformation is clearly given by a differential operator on~$\CC^2$, and hence this deformation remains a local Lie algebra on~$\CC^2$.

At the level of sheaves, the two-term complex~\eqref{twotermdol} is a Dolbeault resolved version of the usual Koszul resolution of the pushforward of the structure sheaf $\cO_{\CC_{z_1}}$ along the map $i : \CC_{z_1} \hookrightarrow \CC^2$ which is the embedding of $\CC_{z_1}$ at $z_2 = 0$:
\beqn\label{twoterm}
\left(
\begin{array}{cccc}
\ul{-1} & & \ul{0}  \\
\varepsilon \; \cO (\CC^2) \otimes \fg & \xto{z_2 \frac{\partial}{\partial \ep}} & \cO (\CC^2) \otimes \fg
\end{array}
\right) \;\;
\xto{\simeq} \;\;  i_* \cO(\CC_{z_1}) \otimes \fg 
\eeqn
The quasi-isomorphism is the restriction map that takes a holomorphic function on~$\C^2$ to its restriction along~$\C_{z_1}$; an explicit quasi-inverse is given, for example, by pulling back a holomorphic function on~$\C_{z_1}$ along the obvious projection map $\pi: \C^2 \rightarrow \C_{z_1}$ and placing the result in degree zero.
 
As with any local Lie algebra, we can consider both its sheaf of sections $\cG'$ and its cosheaf of compactly supported sections $\cG'_{c}$. 
Just as in the case of the sheaf of sections, in cohomology there is an isomorphism of graded cosheaves on $\CC^2$:
\beqn\label{homology}
H^*\left(\cG'_{c}\right) \cong i_* H^*\left(\Omega_c^{0,\bu}(\CC_{z_1}, \fg)\right) .
\eeqn
This statement for cosheaves follows formally from the result about sheaves, but only at the level of cohomology. 
We are interested in a cochain level of this localization result---not only at the level of cosheaves of Lie algebras, but at the level of the corresponding factorization algebras. 

In the notation of Theorem \ref{thm:kmlocal} the factorization algebra of study is 
the enveloping factorization algebra $\cF_{0} (\fg) = \UU(\cG')$.
This factorization algebra is trivial away from the $z_1$-plane.


%

\begin{lem} \label{lem: localkm} 
The factorization enveloping algebra $\UU (\cG')$ is trivial away from $\CC_{z_1} \times \{0\} \hookrightarrow \CC^2$. 
\end{lem}
\begin{proof}
This follows from a statement just about cosheaves of Lie algebras.
Indeed, the cosheaf $\cG_{c}'$, when restricted to the large open stratum, is equivalent to the trivial cosheaf:
\[
\left. \cG_{c}' \right|_{\CC^2 \setminus \CC_{z_1}} \simeq 0
\]
To see this, it suffices to notice that restricting to $z_2 \neq 0$ amounts to inverting $z_2$ in the ring of holomorphic functions on~$\C^2$, over which $\Omega^{0,\bu}(\C^2)$ is a module. Multiplication by an invertible element acts by an isomorphism on the module, so that the complex~\eqref{twoterm} is obviously acyclic after localization at~$z_2$.
\end{proof}

Next, we would like to characterize the factorization algebra $\UU(\cG')$ on the stratum $\CC_{z_1}$. 
As explained in the previous section, the general idea is to use an open tubular neighborhood of the small stratum, and then to push forward the restriction of the factorization algebra~$\UU(\cG')$ to~$\C_{z_1}$ along the projection map. In the case at hand, there is already an obvious projection map~$\pi:\C^2 \rightarrow \C_{z_1}$, but we want to emphasize that our considerations likely generalize to arbitrary curves in~$\C^2$, as considered below in~\S\ref{ssec: moredefs}.

We consider the restricted factorization algebra $\UU(\cG')|_{\CC_{z_1}}$ defined by
\[
\pi_* \left(\left. \UU(\cG') \right|_{{\rm Tub}(\CC_{z_1})}  \right)
\]
In this example, it suffices to take the tubular neighborhood to the entire affine space~$\CC^2$. 

Our goal is to find the explicit relationship between this restricted
factorization algebra and the factorization algebra $\clie_\bu \left(\Omega^{0,\bu}_c(\CC_{z_1}, \fg)\right)$. 
In order to do this, we must fix some additional data. 
Let $\rho : \CC^2 \to \CC$ be a smooth function on $\CC^2$ and $U_1 \subset \Bar{U_1} \subset U_2 \subset U$ be open tubular neighborhoods of $\CC_{z_1} \times \{0\}$ satisfying the following two conditions:
\begin{itemize}
\item $\rho|_{U_1} \equiv 1$, and
\item $\rho|_{\CC^2 \setminus U_2} \equiv 0$.
\end{itemize}
We will refer to $\rho$ as a {\em bump function along $z_2 = 0$}; it can be taken to have image in~$[0,1]\subset \C$, but this does not play a role.

Using $\rho$, define the following map of cosheaves of cochain complexes on $\CC_{z_1}$
\[
\begin{array}{ccccc}
s_{\rho} & : & \Omega^{0,\bu}_c(\CC_{z_1}, \fg) & \to & \pi_* \cG_{c}' \\
&&&& \\
& & \alpha & \mapsto & \displaystyle \rho \, \pi^* \alpha - \ep \frac{\dbar(\rho)}{z_2} \wedge \pi^* \alpha .
\end{array}
\]
Note that, by assumption $\dbar(\rho) \equiv 0$ along $z_2 = 0$, so the expression above is well-defined. 

\begin{prop}\label{prop: compact}
For every choice of $\rho$ as above, the map 
\[
s_{\rho} : \Omega^{0,\bu}_c(\CC_{z_1}, \fg) \xto{\simeq} \pi_* \cG_{c}'
\]
is a quasi-isomorphism of cosheaves of cochain complexes on $\CC_{z_1}$.
\end{prop}

\begin{rmk}
One can view $s_{\rho}$ as an approximation to the map which ``pulls back'' a compactly supported Dolbeault form along the map $\pi : \CC^2 \to \CC$.
The first problem is that since $\pi$ is not proper, pulling back does not preserve compact support.
So, in order to make sense of the pulled back map we must weight it with the function $\rho$. 
The second problem arises due to the fact that $\rho$ is not holomorphic, and so the assignment $\alpha \mapsto \rho \;  \pi^* \alpha$ is not compatible with the $\dbar$-operator. 
It is, however, compatible up to a term proportional to $z_2$.
Hence we can add the $\ep$-dependent term to correct this naive assignment to a cochain map.
\end{rmk}

\begin{rmk}
The map $s_{\rho}$ is independent of the bump function $\rho$ up to homotopy. 
Indeed, a different choice of a bump function $\rho'$ will result in homotopy equivalent maps $s_{\rho} \sim s_{\rho'}$. 
\end{rmk}

\begin{proof}[{Proof of Proposition~\ref{prop: compact}}]
First, we check that $s_{\rho}$ is a cochain map. 
Since the statement is independent of the Lie algebra $\fg$, we will assume $\fg = \CC$ is the trivial Lie algebra for this proof.

For $\alpha \in \Omega_c^{0,\bu}(\CC_{z_1})$, note
\[
\begin{array}{ccl}
\displaystyle \left(\dbar + z_2 \frac{\partial}{\partial \ep}\right) (s_\rho(\alpha)) & = &\displaystyle \dbar(\rho) \wedge \pi^* \alpha + \rho \pi^* \dbar(\alpha) - \ep \frac{\dbar(\rho)}{z_2} \wedge \pi^* \dbar(\alpha) - \dbar(\rho) \wedge \pi^* \alpha \\ & & \\ & = & \displaystyle \rho \pi^* \dbar(\alpha) - \ep \frac{\dbar(\rho)}{z_2} \wedge \pi^* \dbar(\alpha)
\end{array}
\]
This is precisely $s_\rho (\dbar \alpha)$, as desired. 

We now compute the cohomology of the cosheaf $\pi_*\cG_{c}'$. 
On an open set $U \subset \CC_{z_1}$, the value of this cosheaf is $\Omega^{0,\bu}_c (U \times \CC_{z_2})$. 
Using Serre duality, we can identify 
\[
\Omega^{2,\bu}(U \times \CC_{z_2})^\vee \cong \Bar{\Omega}^{0,\bu}_c (U \times \CC_{z_2}) [2] .
\]
This leads to an embedding
\[
(\pi_*\cG'_{c})(U) \hookrightarrow \left(\Omega^{2,\bu}(U \times \CC_{z_2})^\vee [\ep] [-2] , \dbar + z_2 \frac{\partial}{\partial \ep} \right) .
\]
Since the operator $\dbar + z_2 \frac{\partial}{\partial \ep}$ is elliptic, we can apply the Atiyah--Bott Lemma \cite{AB} to see that this embedding is a quasi-isomorphism. 

Thus, it suffices to compute the cohomology of
\[
\left(\Omega^{2,\bu}(U \times \CC_{z_2})^\vee [\ep] [-2] , \dbar + z_2 \frac{\partial}{\partial \ep} \right) .
\]
By the $\dbar$-Poincar\'{e} lemma, this is a equivalent to two-term cochain complex
\[
  \left(\Omega^{2,\text{hol}}(U \times \CC_{z_2})^\vee [\ep] [-2] , z_2 \frac{\partial}{\partial \ep} \right) .
\]
where $\Omega^{2,\text{hol}}$ denotes the sheaf of holomorphic sections of the canonical bundle on $\CC^2$.
We recognize this cochain complex as being linear dual to the ordinary Koszul resolution (\ref{twoterm}) of $\Omega^{1,\text{hol}}(U)[-1]$. 
Thus, we can identify the cohomology of $\pi_*\cG_{c}'(U)$ with 
\[
  \Omega^{1,\text{hol}}(U)^\vee [-1] 
\] 
where $\Omega^{1,\text{hol}}(U)$ is the holomorphic sections of the canonical bundle on $U$. 
Finally, by one-dimensional Serre duality on $\CC_{z_1}$ and by applying Atiyah-Bott Lemma again, this is precisely the $\dbar$-cohomology of $\Omega^{0,\bu}_c(U)$, as desired. 
\end{proof}

A simple observation reveals that $s_{\rho}$ is certainly {\em not} compatible with the Lie brackets, hence is not a map of precosheaves of dg Lie algebras. 
However, the failure for $s_{\rho}$ to be compatible with the Lie brackets is exact for the differential. 
In other words, $s_{\rho}$ can be corrected to an {\em $L_\infty$ map} of precosheaves of dg Lie algebras. 
This $L_\infty$-map will be enough to deduce the statement about factorization algebras, as any $L_\infty$ map induces a map on the Chevalley-Eilenberg cochain complexes. 

In what follows, we set $s^{(1)}_\rho = s_{\rho}$.
Define the $2$-ary map of degree $(-1)$:
\[
\begin{array}{ccccl}
s_{\rho}^{(2)} & : &  i_* \Omega^{0,\bu}_c(\CC_{z_1}, \fg) \times  i_* \Omega^{0,\bu}_c(\CC_{z_1}, \fg) & \to & \cG'_{c} [-1] \\
&&&& \\
&& (\alpha, \beta) & \mapsto & \displaystyle \ep \frac{\rho(\rho-1)}{z_2} [\pi^* \alpha , \pi^* \beta ] .
\end{array}
\]
Note that the expression is well-defined since $1 - \rho \equiv 0$ along $z_2 = 0$. 

\begin{prop}\label{prop: Linf1}
The pair of maps $(s^{(1)}_\rho, s^{(2)}_\rho)$ determine an $L_\infty$ quasi-isomorphism of precosheaves of dg Lie algebras on $\CC_{z_1}$:
\[
(s^{(1)}_\rho, s^{(2)}_\rho) : \Omega^{0,\bu}_c(\CC_{z_1}, \fg) \rightsquigarrow \pi_* \cG_{ c}' .
\]
\end{prop}
\begin{proof}
By Proposition~\ref{prop: compact}, all we need to check is that the pair define a $L_\infty$-morphism.
The $L_\infty$ relation we need to check is of the form
\beqn\label{linfrel}
[s_{\rho}^{(1)} (\alpha), s_{\rho}^{(1)} (\beta)] - s_{\rho}^{(1)} ([\alpha, \beta]) = \left(\dbar + z_2 \frac{\partial}{\partial \ep}\right) s_{\rho}^{(2)}(\alpha, \beta) - s^{(2)}_\rho(\dbar \alpha, \beta) - (-1)^{|\alpha|} s^{(2)}_{\rho} (\alpha, \dbar \beta) 
\eeqn
for $\alpha, \beta \in \Omega^{0,\bu}_c(\CC_{z_1}, \fg)$. 
We prove this relation directly. 
For sake of clutter, we omit the pullback along $\pi$ notation: $\alpha \leftrightarrow \pi^* \alpha \in \Omega^{0,\bu}(\CC^2, \fg)$. 

On one hand, the left hand side of (\ref{linfrel}) is
\[
\left[\rho \alpha - \ep \frac{\dbar(\rho)}{z_2} \wedge \alpha, \rho \beta  - \ep \frac{\dbar(\rho)}{z_2} \wedge \beta\right] - \left(\rho [\alpha, \beta] -  \ep \frac{\dbar(\rho)}{z_2} \wedge [\alpha, \beta]\right) .
\]
Combining terms, we see this is equal to
\[
\rho (\rho - 1) [\alpha, \beta] - \ep \frac{\dbar(\rho)(2 \rho-1)}{z_2} \wedge [\alpha, \beta] .
\]

Now, the right hand side of (\ref{linfrel}) is
\[
\left(\dbar + z_2 \frac{\partial}{\partial \ep}\right) \left(\displaystyle \ep \frac{\rho(\rho-1)}{z_2} [\alpha , \beta ] \right) - \ep \frac{\rho(\rho-1)}{z_2} [\dbar \alpha ,  \beta ] - (-1)^{|\alpha|} \ep \frac{\rho(\rho-1)}{z_2} [\alpha ,\dbar \beta ] 
\]
which matches with the left-hand side by inspection.
\end{proof}

\begin{cor}
Let $\pi : \CC^2 \to \CC_{z_1}$ and $\rho$ be as above.
The $L_\infty$ map $(s^{(1)}_\rho, s^{(2)}_\rho)$ of Proposition \ref{prop: Linf1} defines a quasi-isomorphism of factorization algebras on $\CC_{z_1}$:
\[
\clie_\bu (s_{\rho}) :  \clie_\bu \left(\Omega^{0,\bu}_c(\CC_{z_1}, \fg)\right)  \xto{\simeq} \UU(\cG')|_{\CC_{z_1}} .
\]
\end{cor}
\begin{proof}
This is a formal consequence of Proposition \ref{prop: Linf1} and the fact that pushing forward commutes with taking Chevalley--Eilenberg chains.
Indeed, if $f : X \to Y$ is any map and $\cL$ is a local Lie algebra on $X$, then there is a natural isomorphism of cosheaves
\[
\clie_\bu (f_* \cL_c) \xto{\cong} f_* \clie_\bu (\cL_c) .
\]
\end{proof}

\subsubsection{Central extensions}

We now consider the case where we turn on some non-trivial central extension of the deformed local Lie algebra $\cG'$ and consider the full deformed factorization algebra $\cF_{k_{4d}} = \UU_{\d z_1, k_{4d} \kappa, \theta} (\cG')$. 

Recall that $\cG'$ is a deformation of the current algebra $\cG = \cG_2$. 
In~\S\ref{ssec: KMxtns} we introduced classes in the local cohomology of the undeformed algebra $\cG$.
We are interested in the classes 
\[
\phi^{(2)}_{\d z_1} (k_{4d} \kappa) , \phi^{(3)} (\theta) .
\]
Explicitly, these local cocycles are defined by
\begin{align*}
\phi^{(2)}_{\d z_2} (k_{4d} \kappa) : (\alpha, \ep \alpha') & \mapsto \frac{k_{4d}}{(2\pi i)^2} \int_{\CC^2} \kappa( \alpha \partial \alpha') \d z_2  \\
\phi^{(3)} (\theta) : (\alpha_0,\alpha_1,\alpha_2) & \mapsto \frac{1}{(2 \pi i)^2} \int \theta(\alpha_0 \wedge \partial \alpha_1 \wedge \partial \alpha_2) 
\end{align*}
where $\alpha, \alpha', \alpha_i \in \Omega^{0,\bu}(\CC^2 , \fg)$. 

Upon deforming $\cG \rightsquigarrow \cG'$ each of these remain cocycles in the local cohomology of the deformed algebra. 
However, only the first cocycle remains to be nontrivial.

\begin{lem}
\label{lem: thetatrivial}
The local cohomology class of $\phi^{(3)} (\theta)$ is trivial in $\cloc^\bu(\cG')$ for any $\theta$ as above.
\end{lem}
Recall from~\S\ref{ssec: 2d background} that an invariant pairing $\kappa$ also defines a local cocycle on for the $\fg$-valued holomorphic currents on $\CC_{z_1}$. 
This is the local cocycle that gives rise to the ordinary affine algebra. 

We have already checked the level zero version of localization. 
Theorem \ref{thm:kmlocal} follows from the following computation.
 
\begin{prop} 
  \label{prop: KM cocycle}
Under the pull-back along $s_{\rho} = (s_{\rho}^{(1)}, s_{\rho}^{(2)})$ we have
\deq{
s_{\rho}^* \phi_{\d z_1}^{(2)} (k_{4d} \kappa) = - \frac{k_{4d}}{2} \phi_{2d} (\kappa)  = \phi_{2d} (- k_{4d}\kappa/2) .
}
\end{prop}

\begin{rmk}
Recalling that the scale of~$\kappa$ plays the role of the level, this matches with the result in~\cite{BeemEtAl} that $k_{2d} = - k_{4d}/2$.
\end{rmk}

\begin{proof}
For type reasons, only the pullback along the component $s^{(1)}_\rho$ will contribute a nontrivial class in the cohomology of $\Omega^{0,\bu}_c(\CC_{z_1}, \fg)$. 
Let $\alpha,\beta \in \Omega^{0,\bu}_c(\CC_{z_1}, \fg)$, then
\begin{align*}
(s_\rho^{(1)})^* \phi^{(2)}_{\d z_1} (\alpha,\beta) & = \phi_{\kappa}^{4d} (s_\rho^{(1)}(\alpha), s_\rho^{(1)}(\beta)) \\ 
& =  \phi_{\kappa}^{4d} \left(\rho \pi^* \alpha - \ep \frac{\dbar(\rho)}{z_2} \wedge \pi^* \alpha, \rho \pi^* \beta - \ep \frac{\dbar(\rho)}{z_2} \wedge \pi^* \beta\right)  \\
& = - \frac{1}{(2\pi i)^2} \int_{\CC^2} \d z_2 \wedge \kappa \left(\frac{\dbar(\rho)}{z_2} \wedge \pi^* \alpha \wedge \partial \left(\rho \pi^* \beta \right)\right) \\
& = - \frac{1}{(2\pi i)^2} \int_{\CC^2} \d z_2 \frac{\dbar(\rho^2)}{2z_2}  \wedge \kappa \left(\pi^* \alpha \wedge \partial(\pi^* \beta)\right) - \frac{1}{(2\pi i)^2} \int_{\CC^2} \d z_2 \frac{\partial (\rho) \dbar(\rho)}{z_2}  \wedge \kappa \left(\pi^* \alpha \wedge \pi^* \beta \right) \\
& = - \frac{1}{4 \pi i} \int_{\CC_{z_1}} \kappa (\alpha \partial \beta) \\
& = - \frac{1}{2} \phi_{\kappa}^{2d}(\alpha,\beta) .
\end{align*}
In the fifth line, we have applied Stokes' theorem on an annulus, followed by the residue theorem, in the $z_2$-direction.
But the integral over~$\C_{z_2}$ is also simple to compute by elementary methods, and this is perhaps more illuminating. 
We imagine that our bump function depends only on the radial direction in~$\C_{z_2}$;  that is, $\rho = f(r^2) = f(z_2\zbar_2)$ for some  appropriate function $f$. (The result remains true even if $\rho$ is a more generic bump function.)
It is then easy to see that
\begin{align}
\int_{\C_{z_2}} dz_2\wedge d\bar{z}_2\, \frac{\rho}{z_2}\pdv{\rho}{\bar{z}_2} &= 
\int dz_2 \wedge d\bar{z}_2\, f f'  \\ 
&= - 2i \int_0^\infty \pi d(r^2) \cdot \frac{1}{2} \dv{f^2}{(r^2)}\nonumber \\
&= - \pi i \left. (f^2)\right|^\infty_0 \nonumber \\  
  &= +  \pi i, \nonumber
\end{align}
independent of the choice of~$f$.
\end{proof}

\begin{rmk}
On $\cG$ there are the local cocycles $\phi^{(2)}_{\d z_i}(k_{4d} \kappa)$ for $i=1,2$. 
Upon making our deformation by $z_2 \partial_{\ep}$, only the cocycle $i=1$ remains nontrivial since the deformation has the effect of localizing to the plane $\{z_2=0\}$. 
Without much more difficulty one can generalize this result to a general hyperplane.
Consider the localization of the factorization algebra to the complex plane
\[
P \; : \; \left\{c_1 z_1 + c_2 z_2 = 0\right\} 
\]
which is implemented by the Maurer--Cartan element $c_1 z_1 \partial_\ep + c_2 z_2 \partial_\ep$.  
Denote by ${\rm proj}_P \colon \CC^2 \to P$ the orthogonal projection along $P$.
We can then consider the twisted enveloping factorization algebra of $\cG'$ (where the prime now indicated deformation by this new Maurer--Cartan element) in the presence of the local cocycle
\[
(a_1\phi^{(2)}_{\d z_1} + a_2 \phi^{(2)}_{\d z_2}) (k_{4d} \kappa) .
\]
Upon localizing to the plane $P$, in the same sense as above, one finds the ordinary Kac--Moody factorization algebra supported on $P$ of level
\[
- |{\rm proj}_{P} ({\bf a})| \frac{k_{4d}}{2} \kappa
\]
where ${\bf a} = \begin{pmatrix} a_1 & a_2 \end{pmatrix}$. 
\end{rmk}

\subsection{A deformation of the higher Virasoro algebra}
\label{ssec: defVF}

As above, we deform the local Lie algebra of $\cN=2$ holomorphic vector fields $\fX_2$ by the element $z_2 \frac{\partial}{\partial \ep}$. 
So, consider the local Lie algebra
\deq{
\fX' = \left(\Omega^{0,\bu}(\CC^{2|1}, T\C^{2|1})  , \ \dbar + \left[z_2 \frac{\partial}{\partial \varepsilon}, - \right]\right) .
}
Here $z_2 \partial_\ep$ is acting via the adjoint, or commutator, action.

Consider the map of sheaves
\[
r : \fX \to i_* \left(\Omega^{0,\bu} (\CC_{z_1}, T \CC_{z_1})\right) 
\]
which sends a graded vector field to the restriction of the $z_1$-component to the plane $z_2 = 0$.
That is, if we write a graded vector field as
\[
\ul{\xi} = \xi_1 (z_1,z_2, \varepsilon) \frac{\partial}{\partial z_1} + \xi_2 (z_1,z_2, \varepsilon) \frac{\partial}{\partial z_2} + \xi_\varepsilon (z_1,z_2, \varepsilon) \frac{\partial}{\partial \varepsilon}
\]
then $r(\ul{\xi}) = \xi_1(z_1, z_2 = 0, \varepsilon = 0) \frac{\partial}{\partial z_1}$.
The map $r$ commutes with the Lie bracket with the graded vector field $z_2 \frac{\partial}{\partial \varepsilon}$, so $r$ also defines a map from the deformed $\cN=2$ holomorphic vector fields
\[
r : \fX' \to i_* \left(\Omega^{0,\bu} (\CC_{z_1}, T \CC_{z_1})\right) 
\]
that we denote by the same letter. 

\begin{prop}
  \label{prop: big complex}
Applied to the deformed $\cN=2$ holomorphic vector fields, the map 
\[
r : \fX' \xto{\simeq} i_* \left(\Omega^{0,\bu} (\CC_{z_1}, T \CC_{z_1})\right) 
\]
defines a quasi-isomorphism of sheaves on $\CC^2$.
\end{prop} 
\begin{proof}
  The proof is a simple calculation. We can represent the complex $\fX'_{\N=2}$ by the following diagram:
  \begin{equation}
    \label{eq: big complex}
  \begin{tikzcd}[row sep = 4 em, column sep = 2 em]
\ul{-1} & & \varepsilon \; \Omega^{0,\bu}(\CC^2) \; \frac{\partial}{\partial z_1} \ar[dl, "z_2 \frac{\partial}{\partial \varepsilon}"] &  \varepsilon \; \Omega^{0,\bu}(\CC^2) \; \frac{\partial}{\partial z_2} \ar[dl, "z_2 \frac{\partial}{\partial \varepsilon}"] \ar[dr, "- \iota_{\d z_2} \frac{\partial}{\partial \varepsilon}"] & \\ 
\ul{0} & \Omega^{0,\bu}(\CC^2) \; \frac{\partial}{\partial z_1} & \Omega^{0,\bu}(\CC^2) \; \frac{\partial}{\partial z_2} \ar[dr, "- \iota_{\d z_2} \frac{\partial}{\partial \varepsilon}"] & &\varepsilon \; \Omega^{0,\bu}(\CC^2) \; \frac{\partial}{\partial \varepsilon} \ar[dl, "z_2 \frac{\partial}{\partial \varepsilon}"] \\
\ul{1} & & &  \Omega^{0,\bu} (\CC^2) \; \frac{\partial}{\partial \varepsilon} & 
\end{tikzcd}
\end{equation}
The key observation is that the right quadrilateral forms an acyclic sheaf.
Indeed, both the top right and bottom left diagonal maps are isomorphisms of sheaves of dg vector spaces.
We thus conclude that the deformed sheaf $\fX'$ is quasi-isomorphic to the sheaf
\beqn\label{twoterm2}
\begin{tikzcd}
\varepsilon \; \Omega^{0,\bu}(\CC^2) \; \frac{\partial}{\partial z_1}  \ar[r, "z_2 \frac{\partial}{\partial \ep}"] & \Omega^{0,\bu}(\CC^2) \; \frac{\partial}{\partial z_1} 
\end{tikzcd}
\eeqn
appearing at the top left of~\eqref{eq: big complex} in degrees $-1$ and~$0$. From here, the argument is identical to that in the previous section, since we are once more dealing with the Dolbeault resolution of the Koszul complex representing~$\C_{z_1}$; only the Lie structure is different.
\end{proof}

For the cosheaf version of the deformation, we proceed as we did with the current algebra in the previous section. 
Let $\rho : \CC^2 \to \CC$ be a bump function along $z_2 = 0$ as in~\S\ref{sec: defcur}.
Define the map of cosheaves
\begin{equation}
  \begin{aligned}[c]
  s_{\rho} : \quad
  \Omega^{0,\bu}_c(\CC_{z}, T \CC_{z}) &\to \pi^* \left. \fX_{ c}' \right|_{U} \\[1ex]
 \xi \frac{\partial}{\partial z} &\mapsto \displaystyle \left( \rho \, \pi^* \xi\right) \frac{\partial}{\partial z_1} - \ep \left(\frac{\dbar(\rho)}{z_2} \wedge \pi^* \xi\right) \frac{\partial}{\partial z_1} .
 \end{aligned}
\end{equation}

\begin{prop}
  The map $s_{\rho}$ is a quasi-isomorphism of cosheaves of cochain complexes. It can be corrected to an $L_\infty$ morphism of precosheaves of dg Lie algebras.
\end{prop}
\begin{proof}
  We first check that $s_\rho$ is a cochain map. For simplicity of notation, we omit the pullback symbol $\pi^*$. Observe that
  \deq{
  \begin{aligned}[c]
    \left[ \left( \dbar + z_2 \pdv{ }{\epar} \right) , s_\rho(\xi \partial_z) \right]
  &= \left[ \left( \dbar + z_2 \pdv{ }{\epar} \right) , \left( \rho\xi - \epar \frac{\dbar\rho}{z_2} \wedge \xi \right) \pdv{ }{z_1} \right] \\
  &= \left( \left( \dbar + z_2 \pdv{ }{\epar} \right) \left( \rho\xi - \epar \frac{\dbar\rho}{z_2} \wedge \xi \right) \right) \pdv{ }{z_1},
  \end{aligned}
}
so that the computation reduces to that done for Dolbeault forms in the proof of Proposition~\ref{prop: compact}.

We proceed further by showing that the cohomologies on each side agree. This is sufficient, since $s_\rho$ has an obvious one-sided inverse given by the restriction map. But the argument of Proposition~\ref{prop: big complex} is then enough to reduce the computation of the cohomology in this case to that done for Dolbeault forms in the proof of Proposition~\ref{prop: compact}.

The $L_\infty$ correction term takes a familiar form:
\deq{
  \begin{aligned}[c]
    s_\rho^{(2)} : \quad \Omega^{0,\bu}_c(\C_{z_1},T\C_{z_1}) \otimes \Omega^{0,\bu}_c(\C_{z_1},T\C_{z_1}) &\rightarrow \fX_{c}'[-1] \\[1ex]
    \left( \xi \pdv{ }{z_1}, \lambda \pdv{ }{z_1} \right) &\mapsto \epar \frac{\rho(\rho-1)}{z_2} \left[ \xi \pdv{ }{z_1}, \lambda \pdv{ }{z_1} \right] .
  \end{aligned}
}
The proof that $(s_\rho^{(1)}, s_\rho^{(2)})$ together define an $L_\infty$ map proceeds by a straightforward calculation identical to that given above in the Kac--Moody case; the key fact is that $\dbar$ also obeys a Leibniz rule with respect to the Lie bracket of Dolbeault-valued vector fields.
\end{proof}

\subsubsection{Central extensions}

We now consider the case where we turn on some non-trivial central extension of the deformed local Lie algebra $\fX'$ and consider the full deformed factorization algebra $\cF_{c_{4d}} = \UU_{c_{4d} \psi_2} (\fX')$. 

Given the result for $c_{4d} = 0$ in the last section, this will follow from the analogue of Proposition~\ref{prop: KM cocycle} in the Virasoro case.
We again will find agreement with the result of~\cite{BeemEtAl}.
Recall from Definition~\ref{def: N=2 VF cocycle} above that the relevant cocycle takes the form
\deq{
  \psi_i(\xi, \ep \xi') = \frac{1}{(2\pi i)^2} \int \tr( J\xi) \wedge \partial \tr ( J\xi' ) \wedge dz_i.
}

We have already checked the central charge zero version of localization in the previous subsection. 
Theorem \ref{thm:virlocal} follows from the following computation.

\begin{prop}
  \label{prop: VF cocycle}
  Pulling back along the $L_\infty$ map $s_\rho$, we have that
  \deq{
    s_\rho^* \psi_2 = - \frac{1}{2} \psi_{2d}, \qquad s_\rho^* \psi_1 = 0.
  }
  Accounting for a factor of 24 related to the normalization of~$\psi_{2d}$ and discussed in~\S\ref{ssec: 2d background}, this matches the claim in~\cite{BeemEtAl} that $c_{2d} = - 12 c_{4d}$.
\end{prop}
\begin{proof}
Just as in the previous case, the calculation amounts to computing the pullback of this cohomology class along~$s_\rho$, which can be done as follows: Let $\xi_a(z_1) \partial_1$ be  a Dolbeault-valued vector fields on~$\C_{z_1}$ for $a=1,2$. Then
\begin{align}
  s_\rho^* \psi_i (\xi_1(z_1), \xi_2(z_1) \partial_1) &= \psi_2(s_\rho \xi_1 (z_1) \partial_1, s_\rho \xi_2(z_1)) \nonumber \\ 
&= \psi_i \left( \rho \pi^* \xi_1 \pdv{ }{z_1} - \epar \left( \frac{1}{z_2} \bar\partial \rho \wedge \pi^* \xi_1 \right)\pdv{ }{z_1}, \rho \pi^* \xi_2 \pdv{ }{z_1} - \epar \left( \frac{1}{z_2} \bar\partial \rho \wedge \pi^* \xi_2 \right)\pdv{ }{z_1} \right).
  \end{align}
Setting the arguments equal to $\lambda_a + \epar \lambda_a'$, and omitting the pullback symbol $\pi^*$ for simplicity of notation, we can now directly compute that
  \deq{
    J\lambda_a = \begin{bmatrix}
      L_{\partial_1} (\rho \xi_a)  & 0 \\
    L_{\partial_2} (\rho \xi_a)  & 0 \end{bmatrix}
    = \begin{bmatrix}
    \rho L_{\partial_1} \xi_a & 0 \\
  \dot\rho\xi_a & 0 \end{bmatrix}
  }
  and
  \deq{
    J \lambda'_a = -\begin{bmatrix} 
      L_{\partial_1} \left( \frac{1}{z_2} \bar\partial \rho \wedge \pi^* \xi_a \right) & 0 \\
    L_{\partial_2}  \left( \frac{1}{z_2} \bar\partial \rho \wedge \pi^* \xi_a \right) & 0 \end{bmatrix}
    =
    \begin{bmatrix}
      {z_2}^{-1} \dbar\rho \wedge L_{\partial_1} \xi_a & 0 \\
    \left( {z_2}^{-2} \dbar\rho - {z_2}^{-1} \pdv{ }{z_2} (\dbar\rho) \right)  \wedge \xi_a & 0 
  \end{bmatrix}.
    } 
    Taking traces, applying the $\partial$ operator, and multiplying, we obtain 
    \deq{
      \tr(J \lambda_1) \wedge \partial \tr(J \lambda_2') = 
        \rho L_{\partial_1} \xi_1 \wedge  {z_2}^{-1} \dbar\rho \wedge \partial \left( L_{\partial_1} \xi_2 \right) .
    }
    We now note that $L_{\partial_1} \xi_a = \partial\xi_a/\partial z_1$, so that the cocycle reduces to
    \deq{
      \begin{aligned}[c]
        s_\rho^* \psi_i (\xi_1\partial_z,\xi_2\partial_z) &= \frac{1}{(2\pi i)^2} \int_{\C^2} \frac{1}{2 z_2} \pdv{\xi_1}{z_1} \wedge \dbar(\rho^2) \wedge \partial \left( \pdv{\xi_2}{z_1} \right) \wedge \d{z}_i \\
        &= - \frac{1}{2} \frac{1}{(2\pi i)^2} \int_{\C^2} \left( dz_2 \wedge \frac{\dbar(\rho^2)}{z_2} \right) \left( \pdv{\xi_1}{z_1} \wedge \partial \pdv{\xi_2}{z_1} \right) .
      \end{aligned}
    }
    Since $\xi_a$ is a Dolbeault form on~$\C_{z_1}$, it is clear just by reasons of form degree that $\psi_1$ pulls back to the trivial cocycle.
    Performing the integral over~$\C_{z_2}$ as in Proposition~\ref{prop: KM cocycle} above, we obtain
    \deq{
      s_\rho^* \psi_2 (\xi_1\partial_z,\xi_2\partial_z) = - \frac{1}{2} \frac{1}{2\pi i} \int_{\C_{z_1}} \pdv{\xi_1}{z_1} \wedge \partial \pdv{\xi_2}{z_1} = - 12 \psi_{2d}(\xi_1\partial_z,\xi_2\partial_z) ,
    }
    reproducing precisely the description of the familiar Virasoro cocycle in one complex dimension given in~\cite{Wil-vir} and recalled above in~\S\ref{ssec: 2d background}.
  \end{proof}
  
  \subsection{Exotic deformations of higher symmetry algebras}
\label{ssec: moredefs}

In the preceding subsections, we have shown that the deformation considered by Beem and collaborators (which originates in the global superconformal algebra) appears naturally in our context, taking the form of a Koszul differential, and that their chiral algebras arise from the corresponding deformation of our higher symmetry algebras. However, we wish to emphasize that there are \emph{additional} possible deformations of our algebras, which are not visible at the level  of global superconformal symmetry. While we reserve detailed study of such exotic deformations for future work, we will offer a few remarks below to demonstrate their interest, and will argue in particular that there exist deformations of~$\mathfrak{X}_{2}$ that localize to the holomorphic vector fields on \emph{any} affine algebraic curve in~$\C^2$, and not just to planes. Our remarks are schematic; in particular, we do not here discuss the correct statements at the level of cosheaves. 

Consider the following general setup: Let $A$ denote a commutative differential graded algebra, or more generally a sheaf of such objects. We will ask that $A$ be nonnegatively graded with cohomological differential, and will denote a basis of $\Der(A)$, the degree-zero derivations of~$A$, as a left $A$-module with the symbols $\partial_i$. For example, if $A = \C[z_1,z_2]$, then $\partial_i = \partial/\partial z_1$ or $\partial/\partial z_2$.

We then form the tensor product $A \otimes \C[\epar]$, with $\epar$ an odd variable of degree $-1$. A priori, this is a bigraded cdga,  when  equipped only with the internal differential on~$A$. We are interested in the dg-Lie algebra of its (super) derivations, which was described above in the example of holomorphic vector fields on superspace. As a left $A$-module, graded by $\epar$-degree, we can describe its content with the following table:
\begin{equation}
  \begin{tikzcd}[row sep = 1 ex]
    \ul{-1} & \ul{0} & \ul{1} \\
    & A\cdot \epar \frac{\partial}{\partial \ep} & \\
    A\cdot \epar \partial_i & & A\cdot \frac{\partial}{\partial \ep} \\
    & A\cdot \partial_i &
  \end{tikzcd}
\end{equation}
Now, we ask for deformations of the differential, of total cohomological degree $+1$, that arise from the adjoint action of an element of this dg-Lie algebra on itself. 
Any Maurer--Cartan element gives rise to such a deformation of the differential. 
The simplest class of such elements consist of odd derivations that have vanishing self-bracket and also anticommute with the internal differential on~$A$, so that both terms of the Maurer--Cartan equation are independently zero. In this case, for degree reasons, there are two possible choices:
\begin{itemize}
  \item an element of the form $f \, \frac{\partial}{\partial \ep}$, where $f$ is a closed element of degree zero in~$A$; or
  \item an element of the form $f_i \, \partial_i$, where $f_i$ are degree-one elements of~$A$, such that the result commutes with the internal differential.
\end{itemize}
Both types of deformation are interesting; for example, if $A$ is the Dolbeault complex, we can generate the deformation of the $\bar\partial$ differential to the de~Rham differential by an operator of the second type. However, such deformations have essentially only to do with $A$ itself, and so we will be  interested in the first class of deformations here; these include the deformations made possible by extended superconformal symmetry.

The adjoint action of such an element generates the following differentials (which are maps of left $A$-modules) on our diagram from above:
\begin{equation}
  \begin{tikzcd}
    & A\cdot \epar \frac{\partial}{\partial \ep} \arrow{dr}{f}  & \\
    A\cdot \epar \partial_i \arrow{dr}[swap]{f} \arrow{ur}{\partial_i f} & & A\cdot \frac{\partial}{\partial \ep}\\
    & A\cdot \partial_i \arrow{ur}[swap]{-\partial_i f} &
  \end{tikzcd}
  \label{eq:coho}
\end{equation}

\begin{obs}
  The cohomology of~\eqref{eq:coho}, in $\epar$-degree $-1$, is the left $A$-module
  \deq{
     \bigoplus_i \{ x \in A : f\cdot x = (\partial_i f)\cdot x = 0 \} \cdot \epar \partial_i.
  }
  In particular, when $f$ is not a zerodivisor in~$A$, there is no cohomology in this degree.
  Furthermore, the cohomology in $\epar$-degree $+1$ is the left $A$-module
  \deq{
    A / \langle f, \partial_i f \rangle \cdot \frac{\partial}{\partial \ep} .
  }
Thus, when $A$ is (for example) a polynomial ring in degree zero, the cohomology is precisely the coordinate ring of the singular locus of the affine hypersurface $f=0$, and vanishes when $f$ is a smooth and reduced hypersurface. When, on the other hand, $f = p^n$ for some irreducible (smooth) polynomial $p$, the cohomology will be the quotient of the polynomial ring by the ideal $p^{n-1}$; this is a typical example of behavior in the unreduced case.
\end{obs}

Let us now consider the cohomology in degree zero. We can describe it as the set of elements of the form
\deq{
g_i \partial_i + g_\epar \epar \frac{\partial}{\partial \ep},
}
where the $g$'s are elements of~$A$ and a summation over~$i$ is understood.
These elements are subject to the single relation
\deq{
  g_\epar f = g_i \partial_i f,
}
and are considered modulo the ideal consisting of elements of the form
\deq{
  g_i = f h_i, \quad g_\epar = h_i \partial_i f,
}
which are the image of the differential on elements $h_i \epar \partial_i $ of degree $-1$.

Let us simplify now to the case where $A = \C[z_1,\ldots,z_d]$ is the coordinate ring of affine $d$-space. 
We can analyze the cohomology of stratum by stratum, as we did previously, according to whether we are on the zero locus of~$f$ or in its complement. 
If we assume that $f$ is invertible, it is clear that the cohomology  is trivial. As a sheaf, the cohomology is therefore supported only along the stratum $f=0$. However, if we restrict to this locus (under the  assumption that $f$ is smooth and reduced), it is easy to see that 
the $g_i$ are subject to  the single linear relation $g_i \partial_i f = 0$, so that the vectors appearing in cohomology resolve the tangent sheaf to $f=0$. $g_\epar$ is subject to no relation, but the image of the differential is generated by $\partial_i f$, so that---by the Jacobian criterion for smoothness---it contributes nothing in cohomology. In general, the $g_i$ contribute a copy of the naive tangent space to the hypersurface, and $g_\epar$ contributes a copy of functions on the singular locus, accompanied by $\frac{\partial}{\partial \ep}$.

\subsection{More supersymmetry} 

Instead deforming the $\cN=2$ version of the current and graded algebras of holomorphic vector fields, it is natural to consider similar deformations for the $\cN=4$ versions. 
We do not include the detailed calculation here, but we survey the approach. 

Consider the version of holomorphic vector fields with {\em two} odd directions, which we called $\fX_{3}$.
This is a resolution of the sheaf of holomorphic vector fields on $\CC^{2|2}$. 
For the odd directions introduce the degree $-1$ variables $\ep_1, \ep_2$. 
We consider the deformation
\[
\fX_3' = \left(\Omega^{0,\bu}(\CC^{2|2}, T\C^{2|2})  , \ \dbar + \left[z_2 \frac{\partial}{\partial \varepsilon_2}, - \right]\right) .
\]
From Theorem \ref{thm:gradedvf} and the discussion afterwards, we know that classes in $H^5 (\U(2)) \cong \CC$ give rise to local cocycles of degree $+1$ on $\fX_3$. 
So, up to scale and homotopy there is a unique such local cocycle that we will denote by $\Tilde{\psi}$.
For the present depth of the discussion the explicit form of $\Tilde{\psi}$ will not be used. 

Consider the associated enveloping factorization algebra $\UU_{\Tilde{\psi}} (\fX_3')$ on $\CC^2$ (or $\Sigma \times \CC$). 
As above, this localizes to the $z_1$-plane and defines a holomorphic factorization algebra on $\CC_{z_1}$.
The following result follows from a similar analysis as above. 

\begin{thm}
\label{thm:topvirlocal}
The factorization algebra $\UU_{\Tilde{\psi}} (\fX_3')$ is trivial away from $\{z_2 = 0\} \hookrightarrow \CC^2$. 
The localized factorization algebra is holomorphic and its associated vertex algebra is isomorphic to the $\cN=2$ topological Virasoro vertex algebra.
\end{thm}

Recall that the ordinary Virasoro vertex algebra arises from the factorization enveloping algebra of the sheaf of holomorphic vector fields on $\CC$. 
The $\cN=2$ topological Virasoro vertex algebra arises from the sheaf of holomorphic vector fields on $\CC^{2|1}$. 

The Lie algebra $\fX_3$ is a symmetry of the holomorphic twist of $\cN=4$ supersymmetric Yang--Mills theory. 
Recall that this symmetry enlarges to a symmetry by the local Lie algebra that we denoted $\fX_4^{div}$; this is a resolution for the sheaf of {\em divergence-free} holomorphic vector fields on $\CC^{2|3}$. 
Call the three odd directions $\ep_1,\ep_2,\ep_3$. 
Notice that the graded vector field $z_2 \partial_{\ep_1}$ is divergence-free so it makes sense to deform $\fX_4^{div}$ by this element.
With the proper choice of a twisting cocycle we expect this deformation to localize to a twisted version of the {\em small} $\cN=4$ superconformal vertex algebra.
We leave a proof of this for later work.

\section{Deformations of $\cN=2$ theories}
\label{sec: qft}

In the previous section we have focused on deformations of symmetry algebras present in twists of four dimensional supersymmetric theories. 
We now turn to deformations of four-dimensional quantum field theories themselves from the point of view of the holomorphic twist. 
We choose to focus on the holomorphic twist of theories with $\cN=2$ supersymmetry, and a specific deformation which arises from the $\cN=2$ superconformal algebra. 


Classically, we start with a holomorphic gauge theory on $\CC^2$ which consists of a pure gauge sector and a holomorphic matter (or $\sigma$-model) sector. 
This is the holomorphic twist of $\cN=2$ QCD with hypermultiplets valued in a (complex) symplectic representation, see Proposition~\ref{prop: 2twist}.
In other words, the theory we consider is the holomorphic twist of $\cN=2$ supersymmetric QCD with Lie algebra $\fg$ and matter valued in a representation $V$. 

The free limit of the equations of motion of the theory require that all fields be $\dbar$-closed.
Like in the previous section, we consider deforming this $\dbar$ operator via $\dbar \rightsquigarrow \dbar + z_2 \frac{\partial}{\partial \ep}$.
For a more explicit description of the deformation see Equation (\ref{holdef}) below.
The existence of this deformation is manifest from our results of \S \ref{sec:enhance}.
Indeed, we know by Proposition \ref{prop: symenhance2} that the $\cN=2$ symmetry algebras $\cG_{2}$, $\fX_{2}$ act on the holomorphic twist of any $\cN=2$ theory.
Furthermore, the {\em deformed} symmetry algebras $\cG'_{2}$, $\fX'_{2}$ we considered in \S \ref{sec: defsym} act on this deformation of the holomorphic twist of any $\cN=2$ theory.

We stress that at the classical level, the theory we start with makes sense for any such $\cN=2$ gauge theory, but at the quantum level we find an anomaly in the deformed theory which agrees with the condition that the theory we started with be superconformal.
Specifically, we will show the following.

\begin{prop} 
The holomorphic twist of $\cN=2$ supersymmetric QCD on $\CC^2$ with Lie algebra $\fg$ and matter valued in the symplectic representation $T^*V$, exists at the quantum level. 
There is an anomaly to quantization of $\cN=2$ QCD in the presence of the holomorphic deformation we will introduce in Equation (\ref{holdef}) below.
This anomaly vanishes if and only if
\beqn\label{anomaly}
\Tr_{\fg^{ad}} (X^2) - \Tr_{V}(X^2) = 0
\eeqn
for all $X \in \fg$.
\end{prop}

\begin{rmk}
  This condition can be rewritten in terms of the quadratic Casimir invariant and the dimension of the given representations; it then takes the form
  \deq{
    c_2(\fg)\dim(\fg) = c_2(V)\dim(V).
  }
  For semisimple gauge algebras of type $A$ and matter in fundamental hypermultiplets, this can be rewritten simply using the typical physics parameters $N_f$ and~$N_c$, which indicate gauge algebra $\lie{su}(N_c)$ and matter representation $V = \text{fund}^{\oplus N_f}$. Using familiar expressions for the quadratic Casimir invariants~\cite{Slansky}, the condition becomes 
  \deq{
    N_c (N_c^2-1) = \frac{N_c^2-1}{2N_c}\cdot N_fN_c \quad \Rightarrow N_f = 2 N_c,
  }
  which reproduces the well-known condition for $\N=2$ QCD to be superconformal. One can thus interpret the theorem as indicating that the failure of the original theory to be superconformal is manifested as an anomaly that prevents realization of the higher symmetry algebra at the quantum level.
\end{rmk}

Like the deformed symmetry algebras we met above, the factorization algebra of observables of the deformed theory localizes to the $\CC_{z_1}$ plane. 

\begin{thm}
Suppose the anomaly condition (\ref{anomaly}) is satisfied and let $\Obs(\fg, V)$ be the factorization algebra of quantum observables on $\CC^2$ associated to the holomorphic theory. 
Then, $\Obs(\fg, V)$ is equivalent to a stratified factorization algebra on $\CC^2$, which is trivial away from $\CC_{z_1} \subset \CC^2$, and equivalent to a holomorphic translation invariant factorization algebra $\Obs_{z_1}(\fg, V)$ on $\CC_{z_1}$.
\end{thm}

The final goal is to characterize the factorization algebra $\Obs_{z_1}(\fg,V)$ in a more familiar algebraic description.
By~\cite[Theorem 2.2.1]{CG2}, a holomorphic translation invariant factorization algebra $\cF$ on $\CC$ (satisfying some natural conditions) defines a vertex algebra that we will denote $\VV[\cF]$. 
We then utilize results of~\cite{LiVertex, CG2} which will allow us to relate solutions of the QME, which we have produced by the method of renormalization, and vertex algebras. 
The conclusion is the following.

\begin{prop}
As a vertex algebra, $\VV \left[\Obs_{z_1} (\fg, V)\right]$ is equivalent to the $\fg$-BRST reduction of the $\beta\gamma$ system valued in $V$.  
\end{prop}

\subsection{A holomorphic deformation of $\cN=2$} 

The holomorphic theory we start with is a coupled holomorphic $BF-\beta\gamma$ system, as we introduced in~\S\ref{ssec: theories}. 
We assume the holomorphic BF theory has underlying Lie algebra $\fh = \fg[\ep]$ where $\fg$ is an ordinary Lie algebra\footnote{Taking $\fg$ to be an ordinary Lie algebra as opposed to a dg or $L_\infty$ algebra is for simplicity here, and to match with the familiar situation in the $\cN=2$ untwisted theory.
What is important is that we have this extra odd direction labeled by $\ep$.} and $\ep$ is a parameter of degree $-1$. 
The $\beta\gamma$ system we consider is valued in the graded vector space $\VV = V[\ep]$ where $V$ is a $\fg$-representation, and $\ep$ is as above.

Physically, as we recollected in Proposition \ref{prop: 2twist}, this theory is equivalent to the holomorphic twist of $\cN=2$ supersymmetric QCD with Lie algebra $\fg$ and matter transforming in the symplectic $\fg$-representation $T^*V$. 

The coupled theory is abstractly summarized by thinking about it as the holomorphic BF theory for the semi-direct product graded Lie algebra
\[
\fg_V \overset{\rm def}{=} \fg[\ep] \ltimes V[\ep][-1]
\]
where the semi-direct product is induced by the $\fg$ representation $V$. 
With this notation, the fields of the theory can be written succinctly as 
\[
\Omega^{0,\bu}(\CC^2) \otimes \fg_V [1] \oplus \Omega^{2,\bu}(\CC^2) \otimes \fg_V^* .
\]
In the first component lives the pair of fields $(\ul{A}, \ul{\gamma})$ and in the second component are the conjugate fields $(\ul{B}, \ul{\beta})$. 

The full action can be written as
\beqn\label{N=2super}
S(\ul{A}, \ul{B}, \ul{\gamma}, \ul{\beta}) = \int_{\CC^{2|1}} \<\ul{B} , F_{\ul{A}}\>_\fg + \int_{\CC^{2|1}} \<\ul{\beta}, \dbar_{\ul{A}} \ul{\gamma}\>_V
\eeqn
where $F_{\ul{A}} = \dbar \ul{A} + \frac{1}{2} [\ul{A}, \ul{A}]$ and $\dbar_{\ul{A}} \ul{\gamma} = \dbar \ul{\gamma} + [\ul{A}, \ul{\gamma}]$.
More explicitly, in terms of the components $\ul{\alpha} = \alpha + \ep \alpha'$, we can expand the action as
\begin{align*}
S & = \int_{\CC^2} \<B', \dbar A + \frac{1}{2}[A,A]\>_\fg + \int_{\CC^2} \<B, \dbar A' + [A, A']\>_\fg \\\cL
& + \int_{\CC^2} \<\beta', \dbar \gamma + [A, \gamma]\>_V + \int_{\CC^2} \<\beta, \dbar \gamma' + [A, \gamma'] + [A', \gamma]\>_V
\end{align*}
The first and second lines correspond to the first and second terms in (\ref{N=2super}). 
Note that due to the nature of the pairing between fields and anti-fields, the primed fields $(-)'$ appear precisely once in each term in the action. 

We turn on the following deformation of the holomorphic twist of the free hypermultiplet
\beqn\label{holdef}
I_S(\beta', \gamma') = \int_{\CC^2} z_2 \<B' \wedge A'\>_\fg + \int_{\CC^2} z_2 \<\beta' \wedge \gamma'\>_V .
\eeqn
Equivalently, as an integral over the graded space $\CC^{2|1}$ we can write this action as
\[
I_S(\ul{\beta}, \ul{\gamma}) = \int_{\CC^{2|1}} \<\ul{B} \wedge z_2 \frac{\partial}{\partial \varepsilon} \ul{A}\>_\fg + \int_{\CC^{2|1}} \<\ul{\beta} \wedge z_2 \frac{\partial}{\partial \varepsilon} \ul{\gamma}\>_V .
\]

The deformed theory is completely described by a local Lie algebra that we denote by $\cL(\fg, V)$.
In other words, the Maurer-Cartan elements of $\cL(\fg,V)$ are equivalent to solutions to the classical equations of motion of the deformed theory $S + I_S$. 
The underlying graded Lie algebra is of the form
\beqn\label{sL}
\cL(\fg, V) = \Omega^{0,\bu}(\CC^2) \otimes \fg_V [1] \ltimes \Omega^{2,\bu}(\CC^2) \otimes \fg_V^* .
\eeqn
The differential has two components $\dbar + z_2 \frac{\partial}{\partial \ep}$. 

%
%

\subsection{An exact quantization and the QME}

Holomorphic field theories admit very well-behaved one-loop quantizations in any dimension.
The approach to renormalization for holomorphic theories in the BV formalism that we take is developed in~\cite{BWhol}.
We refer to this work for the notation and conventions used below. 

There are two approaches to producing a renormalized BV action in the case of the deformed holomorphic theory we study here:
\begin{itemize}
\item[(1)] Treat the deformation $z_2 \frac{\partial}{\partial \varepsilon}$ as part of the kinetic term in the action.
This amounts to deforming the linear BV operator
\[
\dbar \rightsquigarrow \dbar + z_2 \frac{\partial}{\partial \varepsilon} .
\]
Since this deformation does not commute with the gauge fixing operator $Q^{GF} = \dbar^*$, the approach of~\cite{BWhol} does not directly apply, and some extra work must be done in producing the renormalized action.
\item[(2)] Consider the deformation as a particular background of the theory. 
This means that we treat the deformation as prescribing a one-parameter family of theories over the ring $\CC[c]$, where the deformed action has the additional interaction term
\[
c \int z_2 \<\beta' \gamma'\> .
\]
In general, treating quadratic terms as deformations of the interacting part of theory is ill-posed since RG flow can produce connected diagrams of infinite size.
Due to the particular form of this deformation, however, the graph expansion is still well-defined even in the presence of this quadratic term. 
\end{itemize}

In principle, by the general formalism to constructing BV theories developed in~\cite{CosRenorm}, both approaches to quantization will yield equivalent results. 
However, one approach may involve significantly more complicated analysis in order to evaluate the respective Feynman diagrams.
We will take approach (2) to studying the quantization of the deformed holomorphic theory, since we can most directly borrow the calculations performed in~\cite{BWhol}.

In doing this, it is convenient to split up the action in the following way:
\beqn\label{theory1}
S + I_S = S_\text{free} + I + I_{S}
\eeqn
where $S_\text{free}$ is the free part of the action in~\eqref{N=2super}, $I$ is the interacting part of the action in~\eqref{N=2super}, and $I_S$ is the deformation in~\eqref{holdef}. 

The gauge fixing condition we choose is given by the operator
\[
Q^{GF} = \dbar^* \otimes 1
\]
which acts on the fields of the theory $\Omega^{0,\bu} (\CC^2) \otimes \fg_V [1]$. 

There is a simple combinatorial observation of the allowable Feynman diagrams that can appear in the graph expansion of the holomorphic theory in the presence of the deformation. 
Without the deformation the theory admits a quantization that is exact at one-loop. 
Even in the presence of the deformation, at one-loop the only possible diagrams that can appear must have external edges labeled by the fields $\ul{A} = A + \ep A'$ or $\ul{\gamma} = \gamma + \ep \gamma'$.
Moreover, since the propagators trade a $\ul{A}$ for a $\ul{B}$ and a $\ul{\gamma}$ for a $\ul{\beta}$, this means that the holomorphic gauge still provides an exact quantization at one-loop. 

The next thing we need to know is that the renormalization group flow acts trivially at one-loop in the presence of the deformation. 
Indeed, by a slight variant of~\cite[Lemma 3.12]{BWhol}, we have the following: 
\begin{lem}
The limit
\[
I[L] + I_S[L] \overset{\rm def}{=} \lim_{\epsilon \to 0} W(P_{\epsilon< L}, I + I_S)
\]
exists. 
Thus, there exists a one-loop finite prequantization of holomorphic theory, even in the presence of the deformation $I_S$. 
\end{lem}
\begin{proof}
The first observation is algebraic. 
Ordinarily, for the weight expansion to be well-defined one must look at graphs with vertices of valence $\geq 3$. 
See~\cite[Chapter 2]{CosBook}. 
The interaction $I_S$ is only quadratic in the fields, but it is nilpotent: $\{I_S, I_S\} = 0$.
Thus, the weight expansion over graphs with bivalent vertices labeled by $I_S$, and trivalent vertices labeled by $I$ is well-defined.

The remainder of the proof is analytic. 
In fact, the proof is nearly identical to the analysis performed in the proof of~\cite[Lemma 3.12]{BWhol}, so we only point out the key additional argument necessary to handle this case. 

For finite $\epsilon$ and~$L$, a general term in the weight of a wheel diagram will be of the form 
\[
\int_{(\CC^2)^k} \left(\prod_{\alpha=1}^k \d z_1^{\alpha} \d z_2^\alpha\right) \Phi(z^1,\ldots, z^k) \left(\prod_{\alpha=1}^k P_{\epsilon<L}(z^{\alpha}, z^{\alpha+1}) z_2^{n_{\alpha}} \right) .
\]
This integral corresponds to taking the weight of a wheel diagram with $k$ vertices. 
Here:
\begin{itemize}
\item $\Phi$ is a compactly supported smooth function on $(\CC^2)^k$;
\item $P_{\epsilon<L}$ is the propagator on $\CC^2$ obtained from the holomorphic gauge fixing condition;
\item $n_{\alpha} \in \{0,1\}$ for $\alpha = 1,\ldots, k$. 
\end{itemize}
For the situation considered in~\cite{BWhol}, it is assumed that the interactions (or vertex labels) are translation invariant; this corresponds to taking $n_\alpha = 0$ for each $\alpha = 1,\ldots, k$ in the above formula. 
In the general case, we simply observe that we can absorb the factors of $z_2^{n_\alpha}$ into the compactly supported function $\Phi$:
\[
\Phi(z^1,\ldots,z^k) \to \Phi'(z_1,\ldots, z_k) = \left(\prod_{\alpha=1}^k z^{n_\alpha}_2\right) \Phi(z^1,\ldots,z^k) .
\]
The new function $\Phi'$ is still compactly supported, and so we can apply an identical analysis carried out in~\cite{BWhol}. 
\end{proof}

In order for the effective family $\{I[L]\}_{L > 0}$ to define a quantum field theory it must satisfy the quantum master equation (QME). 
The renormalized QME exists at each fixed $L > 0$ and is of the form
\[
\dbar I[L] + \frac{1}{2} \{I[L], I[L]\}_L + \hbar \Delta_L I[L] = 0 .
\]
Since our theory is one-loop exact, and satisfies the classical master equation, the only possible anomaly appears at one-loop. 
Thus, if the equation is {\em not} satisfied, then the effective family is said to be anomalous and the scale $L$ anomaly is given by
\[
\Theta[L] = \hbar^{-1} \left(\dbar (I[L]+I_S[L]) + \frac{1}{2} \{I[L]+I_S[L], I[L]+I_S[L]\}_L + \hbar \Delta_L (I[L]+I_S[L]) \right)
\]

By general manipulations of RG flow and the QME, we know that the limit $L \to 0$ of $\Theta[L]$ exists
\[
\Theta = \lim_{L \to 0} \Theta[L] 
\]
Moreover, the functional $\Theta$ is {\em local} and since $\Theta$ is an obstruction, it is also a cocycle.
We now turn to computing this cocycle. 

\subsection{Anomaly cocycle}

The quantization $I[L] + I_S[L]$ is defined as a sum over graphs of genus $\leq 1$. 
It is clear that the anomaly $\Theta[L]$ is also given as a sum over graphs.
In fact, as $L \to 0$, for the holomorphic theories we consider it is shown in~\cite[Proposition 4.4]{BWhol} that this sum concentrates over graphs given by wheels with a particular number of vertices. 

\begin{prop}[{see~\cite[Proposition 4.4]{BWhol}}] \label{prop: anomaly}
The anomaly $\Theta = \lim_{L \to 0} \Theta[L]$ is given as the sum over wheels with precisely three vertices:
\[
\hbar \Theta = \lim_{L \to 0} \lim_{\epsilon \to 0} \sum_{\Gamma \in {\rm Wheel}_{3}, e} W_{\Gamma, e} \left(P_{\epsilon<L}, K_\epsilon, I + I_S\right) .
\] 
Here, the sum is over wheels with $3$ vertices equipped with a distinguished edge $e$.
A general term in the sum is depicted in Figure \ref{fig:wheel}.
\end{prop}

For a wheel $\Gamma$ with distinguished internal edge $e$, the weight $W_{\Gamma,e}(P_{\epsilon<L}, K_\epsilon, I)$ is the graph integral where the heat kernel $K_\epsilon$ is placed on the distinguished edge and the propagators $P_{\epsilon<L}$ are placed on the other internal edges.
The vertices are labeled by $I$ as usual.

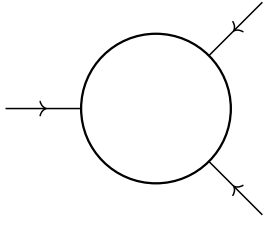
\begin{figure}
\begin{center}
\begin{tikzpicture}[line width=.2mm, scale=1]



		\draw[fill=black] (0,0) circle (1cm);
		\draw[fill=white] (0,0) circle (0.99cm);

		\draw[fermion](180:2) -- (180:1);
		\draw[fermion](45:2) -- (45:1);
		\draw[fermion](-45:2) -- (-45:1cm);
	    	\clip (0,0) circle (1cm);
\end{tikzpicture}
\caption{The anomaly}
\label{fig:wheel}
\end{center}
\end{figure}

The anomaly $\Theta$ is thus given by a sum over weights associated to one-loop wheel diagrams.
By a simple observation on allowable diagrams that can appear, we see that $\Theta$ is only a function of the $\ul{A}$-field.
Thus, it is represented by a cocycle in the local Chevalley--Eilenberg complex
\[
\Theta \in \cloc^\bu(\Omega^{0, \bu}(\CC^2, \fg[\ep])) =  \cloc^\bu(\cG_{\cN=2}) .
\]
We characterized the relevant classes in this local cohomology in~\S\ref{ssec: KMxtns}. 

\begin{prop}
The anomaly cocycle $\Theta$ is a nonzero multiple of the local cocycle $\phi_{\cN=2}^{(2), 2} (\kappa(\fg,V)) \in \cloc^\bu(\cG_{\cN=2})$
where $\kappa(\fg,V)$ is the invariant polynomial
\[
\kappa(\fg,V) = \ch_2^\fg (\fg^{ad}) - \ch_2^\fg(V) \in \Sym^2(\fg^*)^\fg .
\]
In particular, the anomaly vanishes if and only if $\kappa(\fg,V) = 0$.
\end{prop}

\begin{proof}
This is a direct calculation applying the formula for the anomaly given in Proposition \ref{prop: anomaly}. 
We will be short in our calculation of the anomaly, and will emphasize the structural features of the calculation. 

By Proposition \ref{prop: anomaly}, the anomaly is given by evaluating the weight of a wheel where we place the interactions $I$ or $I_S$ on the vertices and the propagator on the edges (and the heat kernel on a distinguished edge). 

Note that for type reasons (since $I_S$ is nilpotent) at most one of the vertices in the $3$-vertex wheel can be labeled by $I_S$, the remaining vertices are labeled by $I$. 
The propagator depends just on the free theory, which has the form $S_{free} = \int \ul{\beta} \dbar \ul{\gamma} + \int \ul{B} \dbar \ul{A}$.
Thus, the propagator splits into two parts:
\[
P = P_{\ul{\beta}\ul{\gamma}} + P_{\ul{B}\ul{A}}
\]

Enumerating the possible $3$-vertex wheels that can appear, we find the following four cases, depicted in Figure \ref{fig: anomaly}:
\begin{itemize}
\item[(I)] All vertices labeled by $I$ and all internal edges labeled by $P_{\ul{\beta\gamma}}$;
\item[(II)] All vertices labeled by $I$ and all internal edges labeled by $P_{\ul{BA}}$;
\item[(III)] Two vertices labeled by $I$, one vertex labeled by $I_S$ and all internal edges labeled by $P_{\ul{\beta\gamma}}$;
\item[(IV)] Two vertices labeled by $I$, one vertex labeled by $I_S$ and all internal edges labeled by $P_{\ul{BA}}$;
\end{itemize}

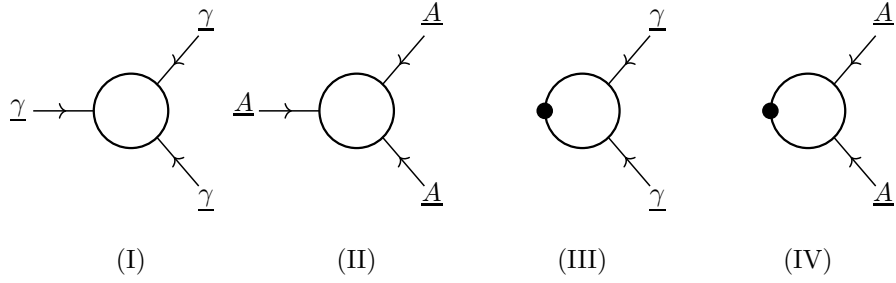
\begin{figure}
\begin{tikzpicture}[line width=.2mm, scale = 1]

		\draw[fermion] (-3.6,1)--(-4.2, 0.3);
		\draw[fermion] (-3.6, -1)--(-4.2, -0.3);
		\draw[fermion] (-5.8,0)--(-5, 0);
		\draw[fill=black] (-4.5,0) circle (.5cm);
		\draw[fill=white] (-4.5,0) circle (.49cm);
		\draw (-3.5, 1.2) node {$\ul{\gamma}$};
		\draw (-3.5, -1.2) node {$\ul{\gamma}$};
		\draw (-6, 0) node {$\ul{\gamma}$};
		\draw(-4.5, -2) node{${\rm (I)}$} ;

		\draw[fermion] (-0.6, 1)--(-1.2, 0.3);
		\draw[fermion] (-0.6, -1)--(-1.2, -0.3);
		\draw[fermion] (-2.8, 0)--(-2, 0);
		\draw[fill=black] (-1.5,0) circle (.5cm);
		\draw[fill=white] (-1.5,0) circle (.49cm);
		\draw (-0.5, 1.2) node {$\ul{A}$}; 
		\draw (-0.5, -1.2) node {$\ul{A}$}; 
		\draw (-3, 0) node {$\ul{A}$}; 
		\draw(-1.5, -2) node{${\rm (II)}$} ;

		\draw[fermion] (2.4, 1)--(1.8, 0.3);
		\draw[fermion] (2.4, -1)--(1.8, -0.3);
		\draw (2.5, 1.2) node {$\ul{\gamma}$}; 
		\draw (2.5, -1.2) node {$\ul{\gamma}$}; 
		\draw(1.5, -2) node{${\rm (III)}$} ;
		\draw[fill=black] (1.5,0) circle (.5cm);
		\draw[fill=white] (1.5,0) circle (.49cm);
		\filldraw[color=black]  (1, 0) circle (.1);

		\draw[fermion] (5.4, 1)--(4.8, 0.3);
		\draw[fermion] (5.4, -1)--(4.8, -0.3);
		\draw (5.5, 1.2) node {$\ul{A}$}; 
		\draw (5.5, -1.2) node {$\ul{A}$}; 
		\draw(4.5, -2) node{${\rm (IV)}$} ;
		\draw[fill=black] (4.5,0) circle (.5cm);
		\draw[fill=white] (4.5,0) circle (.49cm);
		\filldraw[color=black]  (4, 0) circle (.1);
		
\clip (0,0) circle (.3cm);
\end{tikzpicture}
\caption{The anomaly. 
The trivalent vertices are labeled by the cubic interaction $I$.
The bivalent vertices labeled by $\bullet$ are labeled by $I_S$. }
\label{fig: anomaly}
\end{figure}

By general considerations, the anomaly evaluated on $\ul{A} = \alpha \otimes \ul{X} \in \Omega^{0,\bu} (\CC^2) \otimes \fg[\ep]$ where $\alpha$ is a Dolbeault form and $\ul{X} = X + \ep X' \in \fg[\ep]$, 
will have the form
\[
\Theta(\ul{A}) = \Theta^{an}(\ul{\alpha}) \Theta^{alg}(X) .
\]
Here, $\Theta^{an}$ is a local functional of the abelian local Lie algebra $\alpha \in \Omega^{0,\bu}(\CC^2)$ and $\Theta^{alg}$ is an algebraic function of the graded Lie algebra $\ul{X} = X + \ep X' \in \fg[\ep]$. 

We can read off the algebraic factor directly in each of the cases (I)-(IV). 
Note that for type reasons cases (I) and (II) yield functionals that are independent of $\ep$ and hence are just functions of the ordinary Lie algebra $\fg$. 
For the algebraic factor in case (1), the value on an element $X \in \fg$ is
\[
\Tr_{V[\ep]}(X^3) = \Tr_V (X^3) - \Tr_{V} (X^3) = 0 .
\] 
Hence, case (I) does not contribute to the anomaly. 
Similarly, the contribution to the algebraic factor in case (II) is
\[
\Tr_{\fg[\ep]} (X^3) = \Tr_\fg (X^3) - \Tr_{\fg} (X^3) = 0 .
\]
So, case (II) also does not contribute to the anomaly. 

In the last two cases (III), (IV), note that the number of external edges is two (since there is a bivalent vertex). 
Thus the algebraic factor is quadratic as a polynomial on $\fg[\ep]$.
Moreover, it must be linear in $X \in \fg$ and in $\ep X' \in \ep \fg$. 
We can identify such polynomials as quadratic polynomials just on the ordinary Lie algebra $\fg$. 
Doing this, we see that the algebraic factor for case (III) is $\Tr_{V} (X^2)$ and for case (IV) is $- \Tr_{\fg^{ad}}(X^2)$. 
Notice the sign difference since $V$ appears shifted by cohomological degree one relative to $\fg$ in the complex of fields. 

The only thing left to compute is the analytic factor in cases (III) and (IV). 
The analytic factor will again be quadratic, since one of the vertices in bivalent. 
We can therefore assume that we have an abelian Lie algebra, and simply compute the weight of the wheel $\Gamma$ with $3$-vertices where two of the external edges are labeled by elements $\alpha \in \Omega_c^{0,*}(\CC^2)$ and one is labeled by the linear function $z_2$.
In fact, the general formula for the analytic weight of a wheel of this shape for any holomorphic theory on $\CC^2$ has been computed in~\cite[Appendix B]{GWkm} (there, a formula for the weight in any dimension is given).
For general differential form inputs $\alpha, \beta, \gamma$ the formula is a symmetric sum of terms of the form
\[
  \int \alpha_0 \,\partial \alpha_1 \,\partial \alpha_2 .
\]
In our case, we see that the analytic weight is $\int \alpha_0 \,\partial \alpha_1 \,\partial (z_2) = \int \alpha_0 \,\partial \alpha_1 \,\d z_2$ as desired.
\end{proof}

\begin{rmk}
The odd vector field $z_2 \frac{\partial}{\partial \ep}$ that we are deforming the theory by sits inside of the graded Lie algebra of holomorphic vector fields $\fX_{2}$ on $\CC^{2|1}$, see Definition \ref{dfn: vf}.
We argued in~\S\ref{sec: twist} that graded Lie algebra $\fX_{2}$ is the enhancement of the twist of the $\cN=2$ superconformal algebra. 
Moreover, in Proposition \ref{prop: symenhance2} we showed that this enchanced algebra is a classical symmetry of the holomorphic twist of any four-dimensional $\cN=2$ theory on $\RR^4$. 

A more general problem than the one we study in this section is whether we can quantize the symmetry by the full algebra $\fX_{\cN=2}$ acting on the classical theory. 
Of course, we will see the same anomaly as above, but a natural question is whether there are other anomalies. 
If the Lie algebra $\fg$ and the representation $V$ are traceless (that is, $\Tr_{\fg^{ad}}(X) = 0$ and $\Tr_{V} (X) = 0$ for all $X \in \fg$), for instance when $\fg$ is semi-simple, then it turns out that there are no other anomalies.
That is, so long as the condition
\[
0 = \kappa(\fg,V) = \ch_2^\fg (\fg^{ad}) - \ch_2^\fg(V) \in \Sym^2(\fg^*)^\fg 
\]
is satisfied then the full algebra $\fX_{\cN=2}$ is a symmetry of the theory at the quantum level. 
\end{rmk}

We have just computed the anomaly to quantizing the holomorphic theory in the presence of the deformation $I_S$.
If we assume that the anomaly is trivial then we obtain a QFT described by the effective family $\{I[L] + I_S[L]\}_{L > 0}$. 
So long as $\fg$ is semi-simple, this quantization is the unique one-loop exact quantization (up to homotopy) which preserves translation invariance and is $U(2)$-invariant.

By the general formalism of~\cite{CG2}, this QFT defines a factorization algebra of observables which we will denote by $\Obs(\fg, V)$. 
This is a factorization algebra on $\CC^2$ defined over $\CC[[\hbar]]$ whose $\hbar \to 0$ limit is the factorization algebra $\Obs(\fg, V) / \hbar$ which assigns to an open set $U \subset \CC^2$ the cochain complex
\[
\left(\Obs(\fg,V) / \hbar \right)(U) \cong \clie^\bu \left(\cL(\fg,V)(U) \right) 
\]
where $\cL(\fg,V)$ is the local Lie algebra describing the classical theory as introduced in (\ref{sL}). 
In other words $\left(\Obs(\fg, V) / \hbar\right)(U)$ is the cochain complex of classical observables, which are given by functions on the fields supported on $U \subset \CC^2$ equipped with the classical BRST differential.

\subsection{Localization}

The idea of localization is very similar to our analysis of the deformed symmetry factorization algebras in~\S\ref{sec: defsym}. 
We will show that in the presence of the deformation $I_S$, the factorization algebra of observables becomes equivalent to a stratified factorization algebra which is trivial away from the plane $\CC_{z_1}$. 
Along the plane $\CC_{z_1}$, in the next section we will characterize the complex one-dimensional factorization algebra in terms of a vertex algebra. 

Our main tool will be a spectral sequence converging to the cohomology of $\Obs(\fg, V)$, similar to the one considered in~\cite{BPS-SS}.
This filtration also appears in~\cite{EGW} where it is used in the context of 4d $\cN=4$ supersymmetric Yang--Mills theory.
The key property of this spectral sequence is that the first page computes the cohomology of the observables where we turn {\em off} the interactions which are of cubic order and higher.
That is, it is simply the cohomology of the free theory in the presence of the deformation. 
We will find that the cohomology of the free theory localizes to the $\CC_{z_1}$ plane; see Lemma~\ref{lem: bflocal}. 
Upstairs, on $\CC^2$ the spectral sequence converges to the cohomology of the interacting quantum field theory. 
By the fact the the theory localizes at the $E_1$-page, we conclude that the final page of the spectral sequence also localizes to an interacting theory on $\CC_{z_1}$.
Schematically, the picture is the following:
\begin{equation}
\begin{tikzcd}
\{{\rm Free\;theory\; on\;} \CC^2 \} \ar[dd, rightsquigarrow, "{\rm localize}"'] \ar[rrr, Rightarrow] &&& \{{\rm Interacting\;theory\; on\; } \CC^2\} \ar[dd, rightsquigarrow, "{\rm localize}"] \\ \\
\{{\rm Free\;chiral\;theory\;on\;} \CC_{z_1}\} \ar[rrr, Rightarrow] &&& \{{\rm Interacting\;chiral\;theory\;on\;} \CC_{z_1}\}
\end{tikzcd} 
\label{eq: squarepic}
\end{equation}

Now, we get into the proofs of the above assertions. 
As a graded factorization algebra, the $\Obs(\fg,V)$ is given by $\clie^\bu(\cL(\fg,V))[[\hbar]]$, where we recognize $\clie^\bu(\cL(\fg,V))$ is the factorization algebra of classical observables $\Obs(\fg, V)/ \hbar$.
The underlying graded factorization algebra of $\clie^\bu(\cL(\fg, V))$ is of the form
\[
\prod_{n \geq 0} \Sym^n \left(\cL(\fg, V)^\vee[-1]\right) 
\]
where $(-)^\vee$ denotes the continuous linear dual. 
Define the following filtration on $\Obs(\fg, V)$ by
\[
F^p \Obs(\fg, V) = \prod_{2m + n \geq k} \CC \hbar^m \otimes \Sym^n \left(\cL(\fg, V)^\vee[-1]\right)  .
\]

The spectral sequence associated to this filtration has first page given by the cohomology with respect to the {\em linear} part of the differential.
This is the free limit of the classical theory.
The linear term in the differential has two terms: $\dbar + z_2 \frac{\partial}{\partial \ep}$, so the $E_1$-page is given by the following factorization algebra 
\beqn\label{e1}
\cF_1 := H^\bu \left( \Sym \left(\cL(\fg, V)^{\# \vee}[-1]\right) , \dbar + z_2 \frac{\partial}{\partial \ep} \right).
\eeqn
Here, the $\#$ notation $\cL(\fg, V)^\#$ indicates that we are completely forgetting the Lie structure and only remembering the underlying cochain complex. 

At this page, we see the factorization algebra localizes to the $z_1$-plane. 
The proof is completely similar to that of Lemma \ref{lem: localkm}.

\begin{lem}\label{lem: bflocal}
The factorization algebra $\cF_1$ from (\ref{e1}) restricted to $\CC^2 \setminus \CC_{z_1}$ is equivalent to the constant factorization algebra with stalk $\CC$:
\[
\left. \cF_1 \right|_{\CC^2 \setminus \CC_{z_1}} \simeq \ul{\CC} .
\]
\end{lem}
\begin{proof}
It suffices to prove that the sheaf of cochain complexes $\left(\cL(\fg, V)^{\#}, \dbar + z_2 \frac{\partial}{\partial \ep}\right)$ restricted to $\CC^2 \setminus \CC_{z_1}$ is quasi-isomorphic to the trivial sheaf. 
This follows from the familiar short exact sequence (\ref{twoterm}). 
\end{proof}

Just as in~\S\ref{sec: defsym}, we define the factorization algebra $\cF_1'$ on $\CC_{z_1}$ by the pushforward of $\cF_1$ along $\pi : \CC^2 \to \CC_{z_1}$:
\beqn\label{fprime}
\cF'_1 = \pi_* \cF_1 .
\eeqn

The next page in the spectral sequence involves the interacting part of the theory, and its quantization. 
Instead of analyzing the full quantization on $\CC^2$, we will only characterize the quantization of the localized theory on $\CC_{z_1}$. 
This is sensible, by our analysis of the first page in the spectral sequence, since we know the factorization algebra becomes completely trivial away from the $z_1$-plane. 

\subsection{BRST reduction}

To study the quantization of the chiral theory on $\CC_{z_1}$ we make use of an elegant result of~\cite{LiVertex} which sets up a correspondence between quantizations of chiral theories and vertex algebras. 
First, we recall the definition of BRST reduction of a vertex algebra. 

Suppose that $\VV$ is any $\ZZ$-graded conformal\footnote{The same definition holds for quasi-conformal vertex algebra, where we do not demand an action by the full Virasoro, just $\{L_n\}_{n \geq -1}$.} vertex algebra, and a field $J_{\text{BRST}}(z)$ of conformal weight one, cohomological degree one, and has trivial OPE with itself
\[
J_{\text{BRST}}(z) J_{\text{BRST}} (w) \sim 0 .
\]
One then defines the following endomorphism (of cohomological degree one) of the vertex algebra
\[
Q_{\text{BRST}} = \oint \frac{\d z}{2 \pi i} J_{\text{BRST}}(z),
\]
which is called the {\em BRST charge}. 
The condition that $J_{\text{BRST}}(z)$ has trivial OPE with itself implies that $(Q_{\text{BRST}})^2 = Q_{\text{BRST}} \circ Q_{\text{BRST}} = 0$ acting on $\VV$, and hence we can form the complex $\left(\VV, Q_{\text{BRST}}\right)$. 
This object is a {\em dg vertex algebra}. 
Its cohomology
\[
H^* \left(\VV, Q_{\text{BRST}}\right)
\]
is a graded vertex algebra, known as the {\em BRST reduction} of $\VV$ with respect to $J_{\text{BRST}}(z)$. 

\begin{rmk}
The use of terminology is potentially confusing here. In the physics literature, ``BRST'' typically refers to the familiar homological technique for quantizing gauge theories by introducing ghosts, closely connected to the Chevalley--Eilenberg construction. What is called ``BRST reduction'' here is essentially a deformation of the differential, which in most examples imposes the gauge symmetry, but can also be totally unrelated to any Lie algebra action. 
The terminology follows typical usage in the vertex algebra literature; a special case of the procedure is sometimes referred to as ``Drinfeld--Sokolov reduction,'' especially in parts of the literature more closely connected to physics.

Throughout this article, we have used the term ``twist'' to describe precisely the procedure of deforming the differential, but this term is normally restricted to cases where the origin of the deformation is in the action of the physical supersymmetry algebra on the full theory; this is not necessarily the case for the deformations at hand here. 
The physical origin of the BRST reduction at hand lies in passing from the free to the interacting theory, as we have tried to make clear in~\eqref{eq: squarepic} and related discussion above. 
At the four-dimensional level this is, in typical physics usage, neither a BRST nor a twisting differential, but a general deformation of the differential which induces the interaction spectral sequence of~\cite{BPS-SS}.

There is, however, a somewhat askew sense in which BRST is, perhaps, an appropriate name even with respect to physics conventions. Recall that any BRST theory determines a BV theory in canonical fashion, but not all BV theories arise in this fashion.
(Physicists would probably think of such theories as being ones for which the BV formalism can be safely ignored; except for supergravity theories, this is usually the case.) In such a theory, the base of the shifted cotangent bundle is the \emph{BRST theory}, in which antifields are not present; when it carries an internal differential, usually due to the presence of gauge symmetry, this is called the \emph{BRST differential}. 

However, in the twist of such a theory, part of the BRST differential originates in the twisting supercharge. (This is one origin of the overlap in nomenclature.) When supersymmetry is realized off-shell through the use of an auxiliary-field formalism, the twisted theory still arises from a BRST theory, and twisting can be performed just at the BRST level. The auxiliary fields may then be eliminated via their equations of motion. However, after eliminating auxiliary fields, the BRST differential (really, the twisting supercharge) \emph{may depend on interaction terms}---in particular, on superpotential terms---in its action on the component fields. This is the sense in which the introduction of interactions may be thought of as a deformation of the differential, even without passing to the BV formalism (where the action functional is encoded in the BV differential in any case). 

We see no possible choice of nomenclature that does not lead to some confusion or break with tradition, and hope that this remark makes readers sufficiently aware of the existing semantic burden. However, the specific example of two-dimensional BRST reduction we consider below is an example both of a free-to-interacting deformation and of a Chevalley--Eilenberg differential; the theory is the twist of a two-dimensional (0,2) theory with purely gauge interactions. 
\end{rmk}

There is a useful characterization, due to Li \cite{LiVertex}, of the quantum master equation for chiral theories on $\CC$ in terms of vertex algebras.

\begin{thm}[\cite{LiVertex}] \label{thm: Li}
Suppose $E$ is a free chiral theory on $\CC$ with corresponding vertex algebra $\VV[E]$. 
Then, an $\hbar$-dependent field of the vertex algebra $I^{\text{hol}} (z)$ of cohomological degree one satisfies the OPE in $\VV[E]$:
\[
I^{\text{hol}} (z) \cdot I^{\text{hol}}(w) \sim 0
\]
if and only if the corresponding family of functionals 
\[
I[L] = \lim_{\epsilon \to 0} W\left(P_{\epsilon < L}, \int \d z \; I^{\text{hol}}\right)
\]
satisfies the renormalized QME. 
\end{thm}

We see that the condition on $I^{\text{hol}}(z)$ in the theorem above is nearly identical to the condition of the field $J_{\text{BRST}} (z)$ in the general definition of BRST reduction. 
On the other hand, since the resulting family of renormalized functionals $\{I[L]\}$ satisfies the QME, we know by the abstract formalism of~\cite{CG1,CG2} that it defines a quantum field theory and hence a factorization algebra $\Obs_{E,I}$ on $\CC$. 

It is automatic that this factorization algebra is holomorphic and satisfies the conditions of~\cite[Theorem 2.2.1]{CG1}.
Thus, by this theorem, it defines a graded vertex algebra
\[
\VV[\Obs_{E,I}] . 
\]
Combining this with Theorem \ref{thm: Li}, Li dentifies the current $I^{\text{hol}} (z)$ with the standard BRST current $J_{BRST}(z)$ at the level of vertex algebras. 
This is a characterization of the vertex algebra associated to the observables of the quantization of the chiral theory. 


\begin{rmk}
The factorization algebras we consider are all defined over $\CC[\hbar]$. 
When we take the associated vertex algebra we adhere to the convention to specialize $\hbar = 2\pi i$. 
\end{rmk}

We now wish to apply this to the factorization algebra $\cF'_1$ as in (\ref{fprime}) associated to the localized free chiral theory on $\CC_{z_1}$ and the factorization algebra of the resulting chiral deformation obtained from the localization of the interacting theory on $\CC^2$. 

First off, we note that the factorization algebra $\cF'_1$ is equal to the cohomology of a factorization algebra associated to a free chiral theory on $\CC_{z_1}$. 
This is a free chiral theory consisting of a $\fg$-valued ghost $a$, its antifield $b$, and an ordinary $\beta\gamma$ system valued in $V$ whose fields we denote $\gamma_{2d}$ and $\beta_{2d}$ to not confuse them with the higher dimensional $\beta\gamma$ system. 
The action functional of the free chiral theory on $\CC_{z_1}$ is
\[
\int_{\CC_{z_1}} (b \dbar a + \beta_{2d} \dbar \gamma_{2d}) .
\]
The factorization algebra of this free chiral theory will be denoted $\Obs^{free}_{z_1}(\fg,V)$. 
The cohomology of this factorization algebra is precisely the factorization algebra $\cF'_1$. 
The vertex algebra corresponding to $\Obs^{free}_{z_1}(\fg,V)$ is generated by the free fields $a(z), b(z), \gamma_{2d}(z), \beta_{2d}(z)$ has nontrivial OPE's given by
\begin{align*}
a (z) b(w) & \sim \frac{\<a, b\>_{\fg}}{z-w} \\
\gamma_{2d}(z) \beta_{2d}(w) & \sim \frac{\<\gamma, \beta\>_V}{z-w} .
\end{align*}
Denote this vertex algebra by $\VV^{free} [\fg,V]$. 

The spectral sequence with $E_1$-pages $\cF_1$ converges to the cohomology of the factorization algebra $\Obs(\fg,V)$ on $\CC^2$. 
For the factorization algebra on $\CC_{z_1}$ this amounts to taking a further cohomology of $\cF'_1$ which depends on the interacting part of the field theory.  

This can be realized by deforming the free chiral theory $\Obs^{free}_{z_1}(\fg,V)$ by the chiral deformation
\[
I_{2d} = \int_{\CC_{z_1}} \<\beta_{2d}, [a, \gamma_{2d}]\>_V + \<b, [a, a]\>_\fg .
\]
The resulting theory is simply the BF $\beta\gamma$ system on $\CC_{z_1}$. 

By the discussion above, the associated vertex algebra is given by the cohomology of the graded vertex algebra $\VV^{free}[\fg,V]$ with respect to the differential $Q = \oint \d z I^{hol}(z)$:
\[
\VV[\Obs_{z_1}(\fg,V)] = H^\bu \left(\VV^{free}[\fg,V], Q = \oint \<\beta(z), [a(z), \gamma_{2d}(z)]\>_V + \oint \<b(z), [a(z),a(z)]\>_\fg\right) .
\]
This is the description of the BRST reduction of the $\beta\gamma$ system by the affine Kac--Moody Lie algebra generated by the fields $a(z)$, see, for instance, \cite{Karabali}. 

\begin{rmk} 
  The results of this section can be interpreted as a proof, in our formalism, of the descriptions of two-dimensional chiral algebras associated to Lagrangian theories given in~\cite[\S3]{BeemEtAl}, and in particular of the case of $\N=2$ super QCD~\cite[\S5.1--2]{BeemEtAl}. 
\end{rmk}

\printbibliography

\end{document}